\documentclass[runningheads,orivec]{llncs}
\usepackage[T1]{fontenc}
\usepackage{graphicx}
\usepackage{latexsym}
\usepackage{amsmath,amsfonts,mathtools,amssymb}
\DeclarePairedDelimiter\abs{\lvert}{\rvert}
\usepackage{algorithm}
\usepackage{algpseudocode}
\algnewcommand{\LineComment}[1]{\State // #1}
\usepackage[hypertexnames=false]{hyperref}
\usepackage{color}

\usepackage{booktabs}
\usepackage{subcaption}
\usepackage{fullpage}

\newcommand{\NOISE}{{\eta}} % symbol used for noise
\newcommand{\EXPECTED}{{\mathop{\mathbb{E}}}}
\newcommand{\sign}{{\mathop{\mathrm{sign}}}}

\allowdisplaybreaks

\begin{document}

\title{Election Manipulation in Social Networks with Single-Peaked Agents}

\author{Vincenzo Auletta\inst{1} \and Francesco Carbone\inst{1} \and Diodato Ferraioli\inst{1}}

\authorrunning{V. Auletta et al.}
\institute{Università degli Studi di Salerno, Fisciano SA 84084, Italy\\
\email{\{auletta@,f.carbone41@studenti.,dferraioli@\}unisa.it}}
\maketitle              % typeset the header of the contribution

\begin{abstract}
Several elections run in the last years have been characterized by attempts to manipulate the result of the election through the diffusion of fake or malicious news over social networks. This problem has been recognized as a critical issue for the robustness of our democracy. Analyzing and understanding how such manipulations may occur is crucial to the design of effective countermeasures to these practices.

Many studies have observed that, in general, to design an optimal manipulation is usually a computationally hard task. Nevertheless, literature on bribery in voting and election manipulation has frequently observed that most hardness results melt down when one focuses on the setting of (nearly) single-peaked agents, i.e., when each voter has a preferred candidate (usually, the one closer to her own belief) and preferences of remaining candidates are inversely proportional to the distance between the candidate position and the voter's belief. Unfortunately, no such analysis has been done for election manipulations run in social networks.

In this work, we try to close this gap: specifically, we consider a setting for election manipulation that naturally raises (nearly) single-peaked preferences, and we evaluate the complexity of election manipulation problem in this setting: while most of the hardness and approximation results still hold, we will show that single-peaked preferences allow to design simple, efficient and effective heuristics for election manipulation.
\end{abstract}

\section{Introduction}
Nowadays, online social networks
% play a central role in (almost) everyone's life providing quick ways of sharing personal experiences, ideas, and opinions and allowing people to constantly keep in touch. Interestingly, they have also
have
become a ubiquitous, fast, easily accessible source of information: e.g., Matsa and Shearer \cite{matsa_shearer2018} showed that about
% two-thirds of American adults ``at least occasionally get news on social media'' while
one-fifth
of American adults
% often
consults social media to read news.
% , and the top three most visited websites in that sense are Facebook, YouTube, and Twitter.
% Although more than half of the interviewed people are aware of the possible disinformation spreading through social networks (low quality, inaccurate or fake news), convenience, the possibility of having debates and speed make social media an irreplaceable source of information. Considering that
Interestingly a significant part of the interviewed people declared that social media news somehow altered their opinion \cite{matsa_shearer2018}.
% , the topics concerning it are extremely important.
This makes social networks a powerful tool that can be exploited to manipulate people's minds about a particular theme, spreading targeted news to specific users.
Indeed, this spread of information has been apparently exploited in many
% of the
recent elections
% , with notable effects on the electoral outcome
\cite{alcott_gentzkow2017,ferrara2017,bruno2022brexit,giglietto2018mapping}. The most prominent example has been the 2016 U.S. election: in the campaign preceding this event, fake news spreading has been so relevant that some commentators argued that the election's outcome could be different if the campaign had been fair \cite{alcott_gentzkow2017}.
% This was also one of the campaigns involved in the Cambridge Analytica scandal, tightly related to the social network Facebook. In fact, Facebook activities of millions of users were tracked to build their ``psychographic profiles'' \cite{cambridge_analytica_scandal}, allowing political campaign teams to craft user-specific messages supporting the desired candidate.

% Another example of election manipulation is the one regarding the "MacronLeaks disinformation campaign" of the French election in 2017 (\cite{ferrara2017}). In concomitance with the election day, an intense activity of thousands of automated bots was recorded: their aim was to tamper with the election by spreading news on Twitter using fake accounts. Interestingly, \cite{ferrara2017} shows that peaks of the number of tweets generated by the bots were followed by peaks of the number of tweets generated by Twitter accounts driven by real users: this suggests that bots can trigger cascades of disinformation, proving that online social influence is a real threat to democracy, allowing a manipulator to reach a large mass of users and alter the public opinion.

The relevance of the topic leads the AI community to investigate about the problem of manipulating elections by spreading information over social networks. Specifically, the problem has been modelled as follows: let $G=(V, E)$ be a graph representing the (online) social network of the voters, with $V$ being the set of voters and $E$ being the set of (possibly directed) social relationships between voters.
% Let $m$ be the number of candidates in the political election, with candidates $C=\{c_0, c_1, ..., c_{m-1}\}$.
Each voter $v \in V$ has a political opinion that somehow implies particular preferences over the set $C$ of candidates that, in turn, imply a particular vote according to the voting rule that controls the election. The manipulator has a (possibly unlimited) budget $B$ to spend to hire some voters, bribe them, and make them act as influencers to spread some news in favour of or against a target candidate $c^* \in C$. As a result of such influence, some voters (depending on their influenceability and the effectiveness of the hired influencers) will update their opinions and change their votes in favour of or against the target candidate. The aim of the manipulator is to choose the best set of influencers (not violating the budget constraint) to optimize a specific objective function that encodes the chances of victory of the target candidate $c^*$.
Wilder and Vorobeychik \cite{wilder_2018} have been the first to deal with this problem. They indeed prove that it is hard to compute both the set of influencers that maximizes the probability of victory of $c^*$, and the one that optimizes the expected difference between the number of votes of $c^*$ and the number of votes of the best candidate different from $c^*$. However, for the latter problem there is a greedy algorithm that computes a constant approximation of the optimum \cite{wilder_2018}. These results have been extended to more complex settings, focusing, e.g., on different models of information diffusion, different voting rules, and different messages to spread \cite{coro_et_al_2019,uncertainty_paper,castiglioni_et_al}.

These works complement the large literature in AI and social choice about bribery in elections \cite{bartholdi1989,bartholdi1991,bartholdi1992,faliszewski2009_bribery}: they focus on ways of altering the outcome of an election by changing the preference of a few of voters. Anyway, all these works do not take into account the possibility that manipulators could use voters' social relationships to spread the manipulation. Most of the results in these works imply that it is computationally hard to compute the best way to alter an election. Still, most of these hardness results have been showed to melt down when the preferences of voters satisfy the realistic hypothesis of being \emph{single-peaked} or nearly single-peaked \cite{walsh_2007,faliszewski_et_al_2009,brandt_et_al_2010,faliszewski_hemaspaandra_2012}, where single-peakedness implies that candidates can be seen as ordered (e.g., along the political spectrum), voters have a preferred candidate (e.g., the one that is closer to their own political belief) and the preference towards remaining candidates decreases as the distance between their position and the one of the preferred candidate increases.

\paragraph{Our contribution.}
Election manipulation involving information spreading in social networks has not been explicitly studied for the setting in which preferences are single-peaked.
% , and hence we do not know the extent at which known results are weakened by this assumption.
In this work we address this issue, by studying the problem of election manipulation through social influence in single-peaked scenarios. Specifically, we will build over known models of election manipulation in order to embed into them the principles of single-peakedness. Namely, in our model, each voter has an opinion on the topic of the voting and their ranking of alternatives depends on the distance between the candidates' positions and the voter's belief. Here, the diffusion of information has the effect to change the opinion of the voter, and hence it may alter her ranking, but still guaranteeing it to be single-peaked. This model can be also easily extended to encompass nearly single-peaked preferences: these, indeed, may simply arise from voters having a noisy view of candidates position. Given this model, the problem is to find, subject to a budget constraint, the set of ``seeds'' from which to start the information campaign that maximizes the margin of victory of the desired candidate.

It is not hard to check that previous hardness results extend also to this setting. Moreover, we show that there exists an approximation algorithm for the problem guaranteeing to return a set of seeds able to achieve at least a constant fraction of the margin of victory that would be achieved by selecting the optimal set of seeds whenever the target candidate is the one that receives the largest benefit from the campaign\footnote{For example, this may not occur when a message is spread in favour of an extremist party when there are few supporters of an half-extreme party and many supporters for a moderate party: the message causes many votes to move from the moderate towards the half-extreme party, while few votes are conquered by the target party.}. The proposed algorithm is based on a greedy approach, and it is built on Monte Carlo simulations in order to estimate the performance of a seed selection. Unfortunately
%, as highlighted by several works on the topic and confirmed in this work,
this algorithm, even if it guarantees a polynomial time complexity, turns out to be computationally expensive, even for very small instances of the problem.

This motivates the need to design more efficient algorithms, trying to speed up computations while preserving the effectiveness of the manipulation. To this aim, this work proposes and compares several fast heuristics to identify the best voters to influence the electorate; we experimentally show that the best of these heuristics is a variant of the standard PageRank.
We show that the performance of this heuristic overwhelms the one of the approximation algorithm,
% : while the latter handles manipulation problems involving dozens of voters in a few seconds, the former solves the same problems in a few milliseconds,
improving execution times by a factor of (up to) 3000 on average.
% This also allows handling electorates involving thousands of voters in relatively small amount of time.
And this improvement comes with a relatively small loss in terms of effectiveness.
% : {\bf TO CHECK e.g., we will experimentally show that the approximation algorithm is better than the proposed heuristic by only 1\% of the maximum $\Delta MoV$.}
%
Moreover, the proposed heuristics turn out to be robust against altered voters' views of candidates' positions generating only nearly single-peaked preferences: the performances of the heuristics clearly degrade with the amount of noise in the voter's view, but they are very close to the single-peaked case when this noise is limited.
% , especially when the target candidate is not an extremist.

\paragraph{Other Related Works.}
The problem of election manipulation over social networks has been only recently formalized
% by Wilder and Vorobeychik
in
\cite{wilder_2018}. However, several works considered similar issues. E.g.,
% Sina et al.
\cite{sina2015adapting} studies a plurality voting scenario in which the voters can vote iteratively and shows how to modify the relationship among voters to make the desired candidate win an election.
% Auletta et al.
\cite{auletta2015minority,auletta2017information,auletta2017robustness}
% study a majority dynamics
% scenario and
show that in some scenarios, when there are only two candidates, a manipulator controlling the order in which information is disclosed to voters can lead the minority to become a majority.
%Auletta et al.
\cite{auletta2018reasoning} shows that a similar manipulator can lead a bare majority to consensus. These results do not extend to more than two candidates
% , as showed
% by Auletta et al.
% in
\cite{AulettaFFG19,AulettaFG20}.
% Bredereck and Elkind
\cite{bredereck2017manipulating}
shows how this manipulator must select the seeds diffusing information in a two-candidate election.
% Faliszewski et al.
\cite{faliszewski2022opinion} considers a similar issue, but its model does not directly embed the diffusion of information over networks.
Our model for election manipulation is also largely inspired by models of election manipulations under metric preferences \cite{anshelevich2018approximating,wu2022manipulating}.
% of opinion formation. Indeed, one of the most prominent models of opinion formation is undoubtedly the DeGroot (DG) model \cite{degroot74,bindel}, where opinions are continuous and repeatedly updated to the average of the opinions expressed by one’s friends. Several generalization of this model have been proposed: Friedkin-Johnsen \cite{FriedkinJonsen90} suppose that people have an internal belief about the matter in object that limits in some way the influence of friends; other approaches consider discrete opinion spaces \cite{ckoEC13,fgvSAGT12}, or limited/local interactions \cite{fotakis2,fotakis3}.

\section{The Model}
\label{sec:model}
Consider an election with a set of \emph{voters} $V$ and a set of \emph{candidates} (or alternatives) $C=\{c_0, c_1, ..., c_{m-1}\}$. Let $c^* \in C$ be a special \emph{target} candidate such that we want to alter the election in her favour. We consider a \emph{plurality} voting rule: the voters cast a single vote for their preferred candidate (we assume that
% they do not behave strategically, i.e.,
they do not misreport the preferred candidate to alter the election outcome), and the winner of the election is the candidate receiving the largest number of votes.

A candidate $c$ is associated with a position $x_c$ (e.g., their position on the political spectrum). For simplicity, we assume henceforth, that positions are included in $[-1, 1]$.
% In this way, a moderate candidate $c_{\rm mod}$ can be modelled as having a position $x_{c_{\rm mod}}$ close to $0$, a left-wing party $c_{\rm left}$ can be modelled as having a position $x_{c_{\rm left}}$ close to $-1$, and a right-wing party $c_{\rm right}$ can be modelled as having a position $x_{c_{\rm right}}$ close to $+1$.
%
Each voter $v$ also is associated with a position $x_v$ in $[-1, 1]$ reflecting her belief.
The preference of $v$ over candidates depends on her position $x_v$ and her \emph{view} of the candidates' positions. Indeed, we assume that voters may not have a clear picture of the political positions of the parties. For instance, a pure moderate party can be perceived as moderate-left by some voters and moderate-right by others.
% Moreover, the ``amount of confusion'' is different for each voter. As a result, each voter has his own blurred view of the true political spectrum and the political position of the candidates.
% In order t
To model this, we associate with each candidate $c$ a random variable $X_c$ that presumably depends on the true position $x_c$ of the candidate; the blurred view of each voter $v$ consists of a random realization $x_c^v$ of $X_c$. Note that the noisy positions of the candidates in the views are clipped in $[-1, +1]$ to ensure that they remain in the allowed range. Hence, the blurred position of candidate $c$ in the view of voter $v$ can be expressed as
$x_c^v = \left[ x_c + \NOISE(x_c) \right]^{+1}_{-1}$,
where $\NOISE(x_c)$ is the \emph{noise term} depending on the real position of the candidate, and $[\cdot]^{+1}_{-1}$ indicates the clip operation. We assume that $\NOISE(x_{c_i})$  and $\NOISE(x_{c_j})$ are independent, for any $c_i \neq c_j$. We below
consider
% test the performances of our algorithms and heuristics against
several different ways to generate the noise term.

The ranking of voter $v$ with respect to candidate $c$ is then defined with respect to the goal of minimizing the absolute value of the difference between $x_v$ and $x_c^v$: i.e., the most preferred is the one that minimizes $|x_v - x_c^v|$, the second most preferred one achieves the second smallest value of this function, and so on. It is immediate to check that, whenever the view of voters corresponds to real candidates positions, the preferences built in this way are \emph{single-peaked}, i.e., for each voter $v$ there is a preferred candidate $c$, and for each pair of candidate $c', c''$ such that $x_{c'} < x_{c''} \leq x_c$ ($x_{c'} > x_{c''} \geq x_c$), $c'$ is preferred less than $c''$ by $v$. It is easy to see that this method also allows to model nearly single-peaked voters: if the variance of the noise is high, the chances of swapping adjacent candidates on the political spectrum are high, too. Hence, the higher the noise, the higher is the number of swaps necessary to make the resulting ranking single-peaked (see also Figure~\ref{fig:noise_swaps}), that is a very common measure of distance from single-peakedness \cite{erdelyi_2017}.
Anyway, we stress that preferences of voter $v$ are always single-peaked according to her own view, even if they are not single-peaked according to the real position of candidates or other voters' views.
\begin{figure}[H]
\includegraphics[scale=0.6]{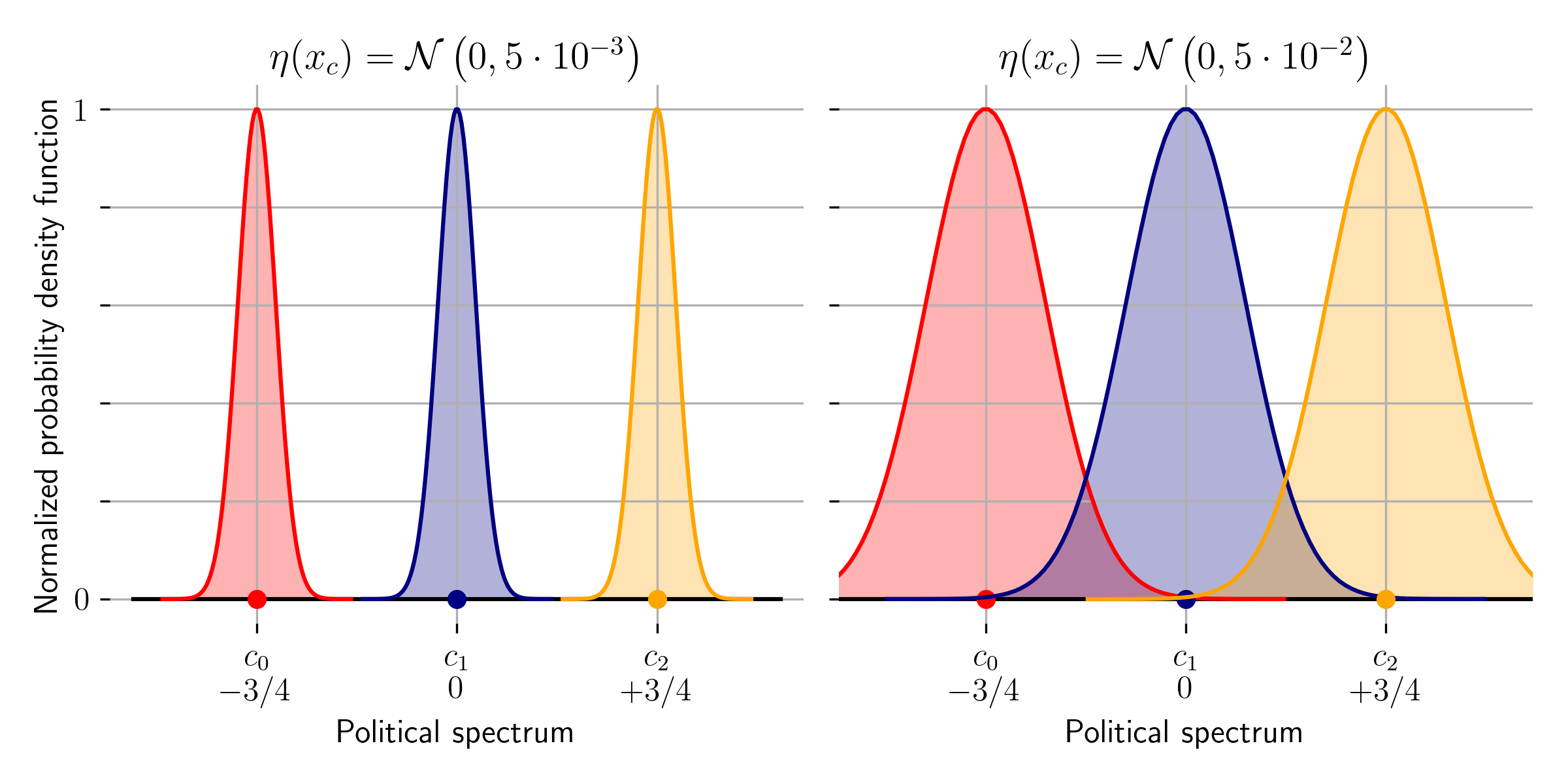}
\caption{An example of the positions of candidates in nearly-single-peaked electorates. The election involves 3 parties $\{c_0, c_1, c_2\}$. The left plot shows the distribution of candidates' positions when noise variance is low. It shows that it is unlikely to swap adjacent candidates in the left-to-right order; hence, preferences are likely to remain single-peaked with respect to the axis $[c_0, c_1, c_2]$. The right plot shows the effect of noise with high variance: the intersections of the distributions can potentially make candidates swap, creating non-single-peaked preferences. The higher the variance of the noise, the higher the chances of swapping adjacent candidates, and the higher the swap distance of the electorate from being perfectly single-peaked.}
\label{fig:noise_swaps}
\end{figure}

A manipulator
% is supposed to interact with voters in order to increase the number of votes of the candidate $c^*$. Specifically, to this aim, the manipulator
can spread information  supporting $c^*$ among voters. Formally, we suppose that voters are arranged on the nodes of a social \emph{network} $G=(V, E, p)$, where $E$ is the set of edges $(u, v)$ connecting voter $u$ to voter $v$, and $p(u,v) \in [0, 1]$ encodes the strength of this relationship, namely how probable is that the information that $u$ sends to $v$ affects the opinion of $v$. The manipulator is then supposed to select a subset $S$ of voters, of size not larger than a given \emph{budget} $B$, from which the information is sent. As most of the previous literature about election manipulation through social networks \cite{wilder_2018,uncertainty_paper,castiglioni_et_al} we assume that information spreads through the network according to the Independent Cascade Model \cite{kempe_et_al}: it starts with $S_0 = S$ and, at each time step $t$, if $S_{t-1}$ is not empty, each voter $u$ in $S_{t-1}$ sends the information to each neighbor $v$ that has not been yet affected, and this neighbor $v$ is affected, and hence inserted in $S_t$, with probability $p(u,v)$.
When a voter $v$ is affected by the news spread by the manipulator (i.e., $v$ belongs to $S_t$ for some $t \geq 0$), his belief is updated. Specifically, the voter's position is moved by a constant amount $\delta$ towards the position (in her view) of the target candidate. If the voter is closer than $\delta$ to the position of $c^*$, then she simply moves to $x_{c^*}^v$. Formally, the voter's new position $\hat{x}_v$ is
% \begin{equation}
    $\hat{x}_v = x_v + \min(\delta, |x^v_{c^*} - x_v|) \cdot \sign(x^v_{c^*} - x_v)$.
% \label{eq:voter_manipulation}
% \end{equation}
% Figure \ref{fig:manipulation_model} pictorially shows the effect of manipulation on the preference of a voter $v$.
% \begin{figure}
%     \centering
%     \includegraphics[scale=0.4]{manipulation_model.png}
%     \caption{An example of the preferences of a voter for the proposed manipulation model.}
% %     The election involves 5 candidates $\{c_0, c_1, c_2, c_3, c_4\}$ equally spaced on the political spectrum. The left plot shows the initial preferences for the voter ($x_v=0.2$). The right plot shows the preferences after the manipulator sent positive news about $c_3$. Assuming the voter is influenced and $\delta=0.2$, $\hat{x}_v=0.4$ and the voter changes his preferred candidate: the manipulation succeeded in altering the voter's opinion in favour of $c_3$. Since the model intrinsically admits single-peaked electorates, preferences are single-peaked with respect to the axis $[c_0, c_1, c_2, c_3, c_4]$ both before and after the manipulation.}
%     \label{fig:manipulation_model}
% \end{figure}

As in previous literature \cite{wilder_2018,uncertainty_paper,coro_et_al_2019,castiglioni_et_al}, we assume that the goal of the manipulator is to choose the set of seed $S$ of size at most $B$ that maximizes the increment in the margin of victory of $c^*$. Specifically, the goal of the manipulator is to maximize the expected change of margin of victory
% $\Delta MoV(S)$, defined as follows:
%\begin{equation}
% \label{eq:DeltaMoV}
$\Delta MoV(S) =  |V^*_{c^*}| - \max_{c \neq c^*}|V^*_{c}| - \left(  |V_{c^*}| - \max_{c \neq c^*}|V_{c}|\right)$,
%\end{equation}
where, by $|V_c|$ and $|V_c^*|$, we mean, respectively, the number of votes for the candidate $c$ before and after the manipulation. Essentially
% \eqref{eq:DeltaMoV} defines
$\Delta MoV$ is the increase of the advantage of $c^*$ over its best opponent before and after the manipulation (that is guaranteed to be always non-negative).
%
% Summarizing, we consider the following problem.
% \begin{definition}[Election manipulation problem]
% \label{def:election_manipulation_problem}
%     Let $V$ be the set of voters organized in a social network $G=(V, E, w)$.
% %     For any $S\subseteq V$, let $\Delta MoV(S)$ be the change of the margin of victory when the nodes in $S$ act as seeds of the information in favour of $c^*$ spread over $G$ according to the Independent Cascade model, and voters update their belief as showed above.
% %     \eqref{eq:voter_manipulation}.
%     Given a limited budget $B$, find the set of seed nodes $S^*$ of size at most $B$ that maximizes the expected value of $\Delta MoV(\cdot)$, i.e.,
% %     \begin{equation}
% %     \label{eq:opt}
%         $S^* = \arg\max_{S} \EXPECTED [\Delta MoV (S)] \colon |S| \leq B$.
% %     \end{equation}
% \end{definition}
Note that the manipulator knows exactly the real position of candidates and of the voters, but she does not know the voters'views.
% how voters view the candidates' positions.

It is not hard to see that by considering the special case of zero-noise, only two candidates and $\delta$ large enough to guarantee that the least preferred candidate becomes the most preferred candidate for each voter $v$ activated by the spread of information, our model reduces to the one considered
% by Wilder \& Vorobeychik
in \cite{wilder_2018}. Hence, the hardness result described there for the election manipulation problem clearly extends to our model. For this reason, in the rest of this work we only look for algorithms able to approximate the optimal choice of the manipulator. Specifically, we say that an algorithm is \emph{$\alpha$-approximate}, for $\alpha \leq 1$ if it always returns a set of seeds $S$ such that $\EXPECTED[\Delta MoV(S)] \geq \alpha \EXPECTED[\Delta MoV(S^*)]$, where $S^* = \arg\max_{S} \EXPECTED [\Delta MoV (S)] \colon |S| \leq B$ is the optimal seed.
% as defined in \eqref{eq:opt}.
% In particular, we will show that, despite our model is much more complex than the one by Wilder \& Vorobeychik \cite{wilder_2018}, we are able to design an algorithm with the same approximation guarantee proved for this simpler model.

In this work we will also consider an extension of previous models: we allow the manipulator to run a multi-round campaign, by choosing in each round the seeds from which the information spreads, and the electorate evolves accordingly.
% We will run such a campaign until the majority of voters votes for the desired candidate (or a maximum round has been met if the desired candidate never wins).
% We will show that, while the algorithms that we propose hardly convince the majority of people to vote for $c^*$, it is usually able to reach this goal in few rounds.

We next introduce two tools that will turn out to be particularly useful in the design of our algorithms.
The first one consists in an algorithm for finding an approximation to the optimal solution for the \emph{Influence Maximization} problem: given a budget $B$, a network $G=(V, E)$ and a weight $w(v)$ for each vertex $w$, find the subset $S^* \subseteq V$ of size at most $B$ that maximizes the expected total weight of vertices in $A(S^*)$, that is the vertices affected by the information sent from $S^*$ and spread according to the Independent Cascade model, i.e., $S^* = \arg\max_{S} \EXPECTED\left[\sum_{v \in A(S)} w(v)\right]$. It is well-known (cf. \cite{kempe_et_al}) that a simple hill-climbing algorithm equipped with Monte Carlo simulations for the estimation of the expectation of random variables is able to return a set of seeds $S$ that provides a $\left(1 - \frac{1}{e} - \varepsilon\right)$-approximation of the expected influence of the optimal choice of seeds $S^*$.

The second useful tool is given by the \emph{PageRank} measure \cite{page_brin}. It is a measure of importance of nodes, that evaluates the importance of a node with respect to the importance of its neighbors. Specifically, the page rank $r(v)$ of a node $v$ in a graph $G=(V,E)$ is computed as $\frac{s}{|N_v|}\sum_{u \in N_v} r(u) + \frac{1-s}{|V|} \sum_{u \in V} r(u)$, where $N_v$ is the set of neighbors of $v$ in $G$, and $s$ is a factor in $[0, 1]$ weighting these two contributions. It is well-known that by repeatedly applying this update rule we will eventually converge to fixed point values (that is, values that do not change after a further application of the update rule), that we will take as the PageRank measure of these nodes.

% Throughout the rest of the paper we assume the reader is familiar with two tools that will turn out to be particularly useful in the design of our algorithms: \emph{Influence Maximization} approximation algorithms, and \emph{PageRank} measure. Interested readers may refer to Appendix~\ref{apx:tools} for more details.

\section{Approximation Algorithm}
\label{sec:greedy}
We here propose a greedy algorithm, that returns a constant approximation of the optimal solution to the Election Manipulation problem in the setting described above whenever the view of voters corresponds to the real position of candidates.
The design of the algorithm directly mimics the ones proposed
% by Wilder \& Vorobeychik (
in \cite{wilder_2018,castiglioni_et_al}.
% ) and Castiglioni et al. (
% \cite{castiglioni_et_al}.
% ).
% In both papers, t
The crux of the algorithm is identifying the voters that, if influenced, will change their minds and vote for the target candidate $c^*$. Then, the problem is solved by simply computing the set of seeds that maximizes the weighted influence maximization \cite{kempe_et_al} for which weights 1 are assigned to such voters and 0 to all the other nodes. Indeed, nodes that already supports the target candidate can be useful for spreading information, but reaching them will not modify the margin of victory of the target candidate; similarly, influencing a voter that cannot be convinced to vote for the target candidate apparently cannot improve the margin of victory.
% The proposed algorithm follows the same ideas.
Unfortunately, influencing a voter that will never vote for $c^*$ after the manipulation is not that pointless: even if the number of votes of the target candidate does not increase, this voter may change her preferred candidate and erode votes for the best opponent of the target candidate, possibly increasing the margin of victory of $c^*$. So we need to prove that, even if this strategy is not accounted by our algorithm (i.e., it only focuses on influencing nodes that can be made to support the target candidate but do not actually do), we still achieve a constant approximation whenever the campaign for the target candidate does not advantage other candidates more than the target itself\footnote{In
% the setting considered by Wilder \& Vorobeychik
\cite{wilder_2018}, it is not necessary to take into account this case since only two candidates are considered. In
% the setting considered by Castiglioni et al.
\cite{castiglioni_et_al}, instead, this case is considered, but the algorithm is showed to provide a constant approximation on stronger assumptions than in our setting.}.

% Specifically, our proposed approach is described in Algorithm \ref{alg:approximation_multiround}.
% Note that we use the algorithm for solving the weighted influence maximization problem cited above as a black box in our procedure.
% %%%%%%%%%%%%%%%%% MoV constructive multi-round %%%%%%%%%%%%%%%%
% \begin{algorithm}
% \caption{Approximation algorithm for the proposed model}
% \label{alg:approximation_multiround}
% \begin{algorithmic}
% \Require the graph $G=(V, E)$, the budget $B$, the displacement $\delta$
%
% \State Let $M$ be the set of voters that do not vote for $c^*$ before the manipulation but will if influenced.
%
% % \LineComment{Define node weights}
% \For {$v$ in $V$}
%     \If{$v \in M$}
%         \State $w(v) = 1$
%     \Else
%         \State $w(v) = 0$
%     \EndIf
% \EndFor
%
% \State $ \hat{S}= weighted\_influence\_maximization(G, w, B)$
%
% \State \textbf{return} $\hat{S}$
% \end{algorithmic}
% \end{algorithm}
Since our approximation algorithm relies on the influence maximization algorithm discussed above, it is easy to see that its computational complexity is $O\left( B \cdot |V| \cdot MI \right)$, where $B$ is the budget, $|V|$ is the number of nodes in the graph, and $MI$ is the cost of estimating the marginal influence of a node via Monte Carlo simulations, that depends on the margin of error that one is willing to accept in the computed estimation. Proposition \ref{proposition_weightedinfluencemaximizationzeroone} proves that $MI$ also polynomially depends on the size of the input, allowing us to conclude that our algorithm is polynomial.

\begin{proposition}
\label{proposition_weightedinfluencemaximizationzeroone}
Let $P$ be a weighted influence maximization problem with the following properties:
\begin{itemize}
\item $n$ is the number of nodes of the graph;
\item for $N \leq n$ nodes $v_1, ..., v_N$, the weights are $w(v_1)=...=w(v_N)=1$;
\item for $n-N$ nodes $v_{N+1},...,v_n$ the weights are $w(v_{N+1})=...=w(v_n)=0$.
\end{itemize}
That is, $P$ is an influence maximization problem in which nodes $v_{N+1},...,v_n$ can be ignored in the final result.

Let $A$ be the set of seed nodes that start the diffusion process. Assume that $A$ contains at least one node $v$ such that $w(v)=1$. Let $\sigma_w(A)$ be the expected sum of weights of influenced nodes at the end of the diffusion dynamic (i.e., the expected number of infected nodes ignoring the nodes with weight 0).

If the diffusion process starting from $A$ is simulated independently at least
$\frac{N^2}{\epsilon^2} \cdot ln \left( \frac{1}{\lambda}\right)$ times, then the average number of influenced nodes with weight 1 over these simulations is a $(1 \pm \epsilon)$-approximation to $\sigma_w(A)$, with probability at least $1-2\lambda^2$.
\end{proposition}
\begin{proof}
Assume that the diffusion process is repeated $T \geq \frac{N^2}{\epsilon^2} \cdot ln \left( \frac{1}{\lambda}\right)$ times.

Let $X_1, X_2, ..., X_T \in [0,1]$ be the fraction of influenced nodes with weight 1 in each of these runs. The estimate of  $\sigma_w(A)$ is

\begin{equation}
    \hat{\sigma}_w(A) = \frac{N}{T} \sum_{i=1}^T X_i \nonumber
\end{equation}

Writing $X=\sum_{i=1}^T X_i$, then $E[X] = \frac{T}{N} \sigma_w(A)$.
Standard bounds (e.g., Theorem 2.3 from \cite{COLIN_MCDIARMID}) give that

\begin{equation}
\label{equation_standardbound_zerooneweightedinfluence}
    Pr\left(\abs{X - \frac{T}{N} \sigma_w(A)} \geq T\gamma \right) \leq 2e^{-2T \gamma^2} \forall \gamma \geq 0
\end{equation}

Considering $\gamma = \frac{\epsilon}{N} \sigma_w(A)$, rearranging the left side of inequality \ref{equation_standardbound_zerooneweightedinfluence} we get

\begin{equation}
    Pr\left(\abs{X\cdot \frac{N}{T} - \sigma_w(A)}\geq \epsilon \sigma_w (A) \right)     \nonumber
\end{equation}

and rearranging the right side of inequality \ref{equation_standardbound_zerooneweightedinfluence} we get

\begin{equation}
    2e^{-2T \gamma^2} \leq 2e^{-2 ln \left(\frac{1}{\lambda}\right) \sigma^2_w (A)} \leq  2e^{-2 ln \left(\frac{1}{\lambda}\right)} = 2\lambda^2 \nonumber
\end{equation}

where the last inequality holds if $\sigma_w(A) \geq 1$, which is true since $A$ contains at least one node $v$ such that $w(v)=1$.\qed
\end{proof}

Moreover, observe that the set of influencers returned by our algorithm has a size that does not exceed the budget, and hence feasible. We next show that it returns a constant approximation of the optimal seed set whenever the view of all voters coincides with candidates' real positions and the campaign for the target candidate does not advantage more other candidates than the target itself.
Specifically, let $X(S)$ be the expected maximum -- among all candidates $c \neq c^*$ -- of the number of voters that do not vote for $c$ and they will do after receiving a message supporting candidate $c^*$ starting from nodes in $S$ (where the expectation is taken over the probabilities of receiving this message). Then we have the following theorem.
% , whose proof is deferred to Appendix~\ref{apx:proof_apx}.
\begin{theorem}
 \label{th:apx}
 The set of influencers $\hat{S} \subseteq V$ returned by our algorithm is such that $\EXPECTED[\Delta MoV (\hat{S})] + X(\hat{S}) \geq \frac{1}{3}\cdot \left( 1- \frac{1}{e}\right) \cdot max_S \EXPECTED[\Delta MoV (S)]$ whenever $x_c^v = x_c$ for each $v \in V$ and each $c \in C$.
 Hence, our algorithm returns a constant approximation whenever $X(\hat{S}) = O(\EXPECTED[\Delta MoV (\hat{S})])$.
%  , i.e., the benefit that candidates different from $c^*$ may achieve from a message supporting $c^*$ sent from $\hat{S}$ is not much larger than the advantage received by $c^*$ itself.
 \end{theorem}
\begin{proof}
%     This proof mimics the one of Theorem 5.2 in Wilder \& Vorobeychik's paper (\cite{wilder_2018}); hence, new concepts are needed.
    Given $G=(V, E, p)$, the corresponding live graph is a graph $G'$ over $V$ in which each edge $e \in E$ belongs to $G'$ with probability $p(e)$. Observe that there are $2^{|E|}$ possible live graphs, and let us denote with $\mathcal{G}$ the set containing all these live graphs. Moreover we can associate a distribution $\mathcal{P} \colon \mathcal{G} \rightarrow [0, 1]$ to these live graphs, with $\mathcal{P}(G_i) = \prod_{e \in G'} p(e) \prod_{e \notin G'} (1-p(e))$ for each $G_i \in \mathcal{G}$. Given a live graph $G_i$ and a seed set $S$, we denote with $f(S, i)$ the number of nodes reachable from the seed set $S$ in the live graph $G_i$. It is immediate to see that the expected number of influenced nodes in the graph, starting the diffusion from the seed set $S$, is simply $f(S)=\mathop{\mathbb{E}}_\mathcal{P} \left[ f(S, i) \right]$.

    Recall that $M$ is the set of voters that do not vote for $c^*$ but will do if influenced.
    Let $V_c$ be the set of voters that vote for candidate $c$ before the manipulation, and $\overline{M}_c $ be those voters that do not vote for $c$, but they will vote for $c$ if they receive a message supporting $c^*$, i.e., if they move their belief of at most $\delta$ towards $x_{c^*}^v$.
    Note that $M_c \cap V_{c^*}$ must be empty, since it is impossible to make a voter that prefers $c^*$ to change her mind through a message supporting $c^*$.

    Moreover, let $\hat{S}$ be the seed set returned by our algorithm, and $S^*$ be an optimal solution maximizing the expected change of margin of victory.

    After the manipulation due to the influencers $\hat{S}$, the increment in the margin of victory between $c^*$ and any other candidate $c_i \neq c^*$ can be expressed (considering the diffusion over the graph $G_y$) as $g_{c^*}(\hat{S}, y, c_i)$ and observe that
    \begin{align*}
        & \sum_{v \in M - V_{c_i}} \chi(v, \hat{S}, y) + 2 \sum_{v \in M \cap V_{c_i}} \chi (v, \hat{S}, y) \geq\\
        & g_{c^*}(\hat{S}, y, c_i) \geq \sum_{v \in M - V_{c_i}} \chi(v, \hat{S}, y) + 2 \sum_{v \in M \cap V_{c_i}} \chi (v, \hat{S}, y)\\
        & \qquad - \sum_{v \in \overline{M}_{c_i}} \chi(v, \hat{S}, y)\\
        & \geq \sum_{v \in M - V_{c_i}} \chi(v, \hat{S}, y) + 2 \sum_{v \in M \cap V_{c_i}} \chi (v, \hat{S}, y) - \max_{c \neq c^*} \sum_{v \in \overline{M}_c} \chi(v, \hat{S}, y),
    \end{align*}
    where $\chi(v, \hat{S}, y)$ is 1 if node $v$ is reachable from nodes in $\hat{S}$ in the live graph $G_y$, else 0. The change of margin of victory in the live graph $G_y$ is then
    \begin{equation}
        \Delta MoV(\hat{S}, y) = \min_{c_j \neq c^*} \left\{g_{c^*}(\hat{S}, y, c_j) - |V_{c_j}| + \max_{c_i \neq c^*} |V_{c_i}| |\right\} \nonumber
    \end{equation}
    and the expected change of margin of victory is $\mathop{\mathbb{E}}_y \left[ \Delta MoV (\hat{S}, y)\right]$. Note that
    $$\sum_{v \in M - V_{c_i}} \chi(v, \hat{S}, y) + 2 \sum_{v \in M \cap V_{c_i}} \chi (v, \hat{S}, y) = f(\hat{S}, y, M) + f(\hat{S}, y, M \cap V_{c_i}),$$
    where $f(S, y, A)$ is the number of nodes in $A$ reachable from $S$ in the live graph $G_y$. Then we have that
    \begin{equation}
    \label{eq:wilder_proof_alternative_deltamov}
    \begin{aligned}
        & f(\hat{S}, y, M) + \min_{c_j \neq c^*} \left[ f(\hat{S}, y, M \cap V_{c_j}) - |V_{c_j}| + \max_{c_i \neq c^*} |V_{c_i}|  \right] \geq\\
        & \Delta MoV(\hat{S}, y) \geq f(\hat{S}, y, M)\\
        & + \min_{c_j \neq c^*} \left[ f(\hat{S}, y, M \cap V_{c_j}) - |V_{c_j}| + \max_{c_i \neq c^*} |V_{c_i}|  \right]  - \max_{c \neq c^*} \sum_{v \in \overline{M}_c} \chi(v, \hat{S}, y).
    \end{aligned}
    \end{equation}

    By definition, our algorithm greedily maximizes the submodular function $\mathop{\mathbb{E}}_y \left[f(\cdot, y, M)\right]$. In fact, it maximizes the expected sum of the weights of the influenced nodes in the derived influence maximization problem; but weights are defined such that the expected sum of influenced nodes is exactly the expected number of influenced nodes in $M$, that is $\mathop{\mathbb{E}}_y \left[f(\cdot, y, M)\right]$. Hence, if $S'= \arg\max_A \mathop{\mathbb{E}}_y \left[f(A, y, M)\right]$, then
    \begin{equation}
    \label{eq:wilder_proof_influence_maximization_bound}
    \begin{aligned}
        \mathop{\mathbb{E}}_y \left[f(\hat{S}, y, M)\right] & \geq \left(1-\frac{1}{e}\right)\cdot \mathop{\mathbb{E}}_y \left[f(S', y, M)\right]\\ & \geq \left(1-\frac{1}{e}\right)\cdot \mathop{\mathbb{E}}_y \left[f(S^*, y, M)\right]
    \end{aligned}
    \end{equation}
    where the first inequality holds because of the guarantees of the algorithm solving weighted influence maximization, whereas the second inequality holds because $S'$ maximizes $\mathop{\mathbb{E}}_y \left[f(\cdot, y, M)\right]$ by definition.

    Let $c(\hat{S}, y)=\arg\min_{c_i} f(\hat{S}, y, M \cap V_{c_i})-|V_{c_i}|$ be the candidate achieving the minimum in the definition of $\Delta MoV(\hat{S}, y)$ in \eqref{eq:wilder_proof_alternative_deltamov}. Note that for any candidate $c_i$, $f(\hat{S}, y, M) \geq f(S, y, M \cap V_{c_i})$, simply because $M \cap V_{c_i} \subseteq M$. Hence, for $S \in \{S^*, \hat{S}\}$
    \begin{equation}
    \label{eq:wilder_proof_second_bound}
        \mathop{\mathbb{E}}_y \left[f(S^*, y, M)\right] \geq \mathop{\mathbb{E}}_y \left[f(S^*, y, M \cap V_{c(S, y)})\right]
    \end{equation}

    By \eqref{eq:wilder_proof_second_bound} we get
    \begin{equation}
    \label{eq:wilder_proof_first_second_third_bound}
    \begin{aligned}
        \mathop{\mathbb{E}}_y \left[f(S^*, y, M)\right] & \geq  \frac{1}{3} \left( \mathop{\mathbb{E}}_y \left[f(S^*, y, M)\right]\right.\\
        & + \mathop{\mathbb{E}}_y \left[f(S^*, y, M \cap V_{c(S^*, y)})\right]\\
        & + \left.  \mathop{\mathbb{E}}_y \left[f(S^*, y, M \cap V_{c(S, y)})\right] \right)
    \end{aligned}
    \end{equation}

    Moreover,
    \begin{equation}
    \label{eq:wilder_proof_mixed_bound}
    \begin{aligned}
        & \mathop{\mathbb{E}}_y \left[ f(\hat{S}, y, M) + f(\hat{S}, y, M \cap V_{c(\hat{S}, y)}) \right] \geq \mathop{\mathbb{E}}_y \left[ f(\hat{S}, y, M)\right]
        \\
        & \geq \frac{1}{3}\left( 1- \frac{1}{e} \right) \cdot \left( \mathop{\mathbb{E}}_y \left[f(S^*, y, M)\right] +
        \mathop{\mathbb{E}}_y \left[f(S^*, y, M \cap V_{c(S^*, y)})\right]\right.\\
        & \qquad \left. +  \mathop{\mathbb{E}}_y \left[f(S^*, y, M \cap V_{c(S, y)})\right]
        \right)
    \end{aligned}
    \end{equation}
    where the first inequality holds because the left hand side includes an extra non-negative term, and the second inequality derives from inequalities \ref{eq:wilder_proof_first_second_third_bound} and \ref{eq:wilder_proof_influence_maximization_bound}.
    Also note that $\mathop{\mathbb{E}}_y\left[\max_{c \neq c^*} \sum_{v \in \overline{M}_c} \chi(v, \hat{S}, y)\right] = X(\hat{S})$.
    We can bound the margin of victory of $\hat{S}$ as follows (for simplicity, we will write ``$\min_{c_j}$'' for ``$\min_{c_j\neq c^*}$'', and the same for ``$\max_{c_j}$''):
    \begin{align*}
        & \mathop{\mathbb{E}}_y \left[ \Delta MoV (\hat{S}, y) \right] + X(\hat{S})\\
        & \geq \mathop{\mathbb{E}}_y\left[ f(\hat{S}, y, M) + \min_{c_j} \left(f(\hat{S}, y, M \cap V_{c_j}) - |V_{c_j}|\right) + \max_{c_i} |V_{c_i}| \right] \tag*{(by \eqref{eq:wilder_proof_alternative_deltamov})}\\
        & = \mathop{\mathbb{E}}_y\left[ f(\hat{S}, y, M) + f(\hat{S}, y, M \cap V_{c(\hat{S}, y)})\right] + \mathop{\mathbb{E}}_y\left[ \max_{c_i} |V_{c_i}| - |V_{c(\hat{S}, y)}|\right] \tag*{(by definition of $c(\hat{S}, y)$)}\\
        & \geq \frac{1}{3}\left( 1- \frac{1}{e} \right) \cdot \mathop{\mathbb{E}}_y\biggl[ f(S^*, y, M) + f(S^*, y, M \cap V_{c(S^*, y)})\\
        & \qquad + f(S^*, y, M\cap V_{c(\hat{S}, y)}) \biggr] + \mathop{\mathbb{E}}_y\left[\max_{c_i} |V_{c_i}| - |V_{c(\hat{S}, y)}|\right] \tag*{(by \eqref{eq:wilder_proof_mixed_bound})} \\
        & \geq \frac{1}{3}\left( 1- \frac{1}{e} \right) \cdot \mathop{\mathbb{E}}_y\biggl[ f(S^*, y, M) + f(S^*, y, M \cap V_{c(S^*, y)})\\
        &\qquad + f(S^*, y, M\cap V_{c(\hat{S}, y)}) + \max_{c_i} |V_{c_i}| -  |V_{c(\hat{S}, y)}|\biggr] \tag*{(1/3(1-1/\emph{e}) multiplies every term)}\\
        & = \frac{1}{3}\left( 1- \frac{1}{e} \right) \cdot \mathop{\mathbb{E}}_y\biggl[ f(S^*, y, M) + f(S^*, y, M \cap V_{c(S^*, y)})\\
        & \qquad + f(S^*, y, M\cap V_{c(\hat{S}, y)}) + \max_{c_i} |V_{c_i}| -  |V_{c(\hat{S}, y)}|+ |V_{c(S^*, y)}|-|V_{c(S^*, y)}|\biggr] \tag*{(add and subtract the same quantity)} \\
        & = \frac{1}{3}\left( 1- \frac{1}{e} \right) \cdot \mathop{\mathbb{E}}_y\biggl[f(S^*, y, M) + \min_{c_j}\left(f(S^*, y, M \cap V_{c_j})-|V_{c_j}|\right)\\
        &\qquad +f(S^*, y, M \cap V_{c(\hat{S}, y)}) + \max_{c_i} |V_{c_i}| - |V_{c(\hat{S}, y)}| + |V_{c(S^*, y)}|\biggr] \tag*{(by definition of $c(S^*,y)$)}\\
        & \geq \frac{1}{3}\left( 1- \frac{1}{e} \right) \cdot \biggl( \mathop{\mathbb{E}}_y[\Delta MoV(S^*, y)] + \mathop{\mathbb{E}}_y\biggl[ f(S^*, y, M \cap V_{c(\hat{S}, y)})\\
        & \qquad + |V_{c(S^*, y)}| - |V_{c(\hat{S}, y)}|\biggr]\biggr) \tag*{(by \eqref{eq:wilder_proof_alternative_deltamov})}
    \end{align*}

    By definition of $c(S^*, y)$, $f(S^*, y, M\cap V_{c(S^*, y)}) - |V_{c(S^*, y)}| \leq f(S^*, y, M\cap V_{c(\hat{S}, y)}) - |V_{c(\hat{S}, y)}| $, and hence
    \begin{equation}
    \label{eq:wilder_proof_last}
        |V_{c(S^*, y)}| - |V_{c(\hat{S}, y)}| \geq f(S^*, y, M \cap V_{c(S^*, y)}) - f(S^*, y, M \cap V_{c(\hat{S}, y)})
    \end{equation}

    We can conclude that
    \begin{align*}
        & \mathop{\mathbb{E}}_y[\Delta MoV(\hat{S}, y)] + X(\hat{S})\\
        & \geq \frac{1}{3}\left( 1- \frac{1}{e} \right) \cdot \biggl( \mathop{\mathbb{E}}_y[\Delta MoV(S^*, y)] + \mathop{\mathbb{E}}_y\biggl[f(S^*, y, M \cap V_{c(\hat{S}, y)})\\
        &  \quad + f(S^*, y, M \cap V_{c(S^*, y)})-f(S^*, y, M \cap V_{c(\hat{S}, y)}) \biggr]\biggr) \tag*{(by \eqref{eq:wilder_proof_last})}\\
        & = \frac{1}{3}\left( 1- \frac{1}{e} \right) \cdot \biggl( \mathop{\mathbb{E}}_y[\Delta MoV(S^*, y)] + \mathop{\mathbb{E}}_y\biggl[f(S^*, y, M \cap V_{c(S^*, y)})\biggr]\biggr) \\
        & \geq \frac{1}{3}\left( 1- \frac{1}{e} \right) \cdot\mathop{\mathbb{E}}_y[\Delta MoV(S^*, y)] \tag*{(since \emph{f} is non-negative)}
    \end{align*}

    Thus if $X(\hat{S}) = O(\mathop{\mathbb{E}}_y[\Delta MoV(\hat{S}, y)])$, then there is a constant $c$ such that $X(\hat{S}) \leq c \cdot \mathop{\mathbb{E}}_y[\Delta MoV(\hat{S}, y)]$, and thus
    $$
     \mathop{\mathbb{E}}_y[\Delta MoV(\hat{S}, y)] \geq \frac{1}{3(c+1)}\left( 1- \frac{1}{e} \right) \cdot\mathop{\mathbb{E}}_y[\Delta MoV(S^*, y)],
    $$
    and thus a constant approximation.
    \qed
\end{proof}

 Note that in case of nearly-single-peaked electorates the algorithm still works, but its performances in terms of $\Delta MoV(\cdot)$ depend on the amount of noise: the higher the noise, the more the inconsistency between reality and blurred views of the voters, the lower the effectiveness of the manipulation since the manipulator estimates wrong weights $w(v)$ for each voter $v \in V$.

\section{The Heuristics}
The approximation algorithm presented in the previous section provides a formal guarantee of the quality of the solution. However, it inherits from the weighted influence maximization algorithm used as black box the computational drawbacks of being computationally expensive, even if it is polynomial in the size of the input.
Specifically, the proposed algorithm requires a large number of simulations  in order to estimate the marginal influence of each node (see Proposition~\ref{proposition_weightedinfluencemaximizationzeroone}).
% indeed shows that this number is $\Omega(\gamma|V|)$, with $\gamma$ being a very huge constant).
Even if faster algorithms have been proposed for the influence maximization algorithm (see, e.g., \cite{borgs2014maximizing}), the influence maximization algorithm is often solved in practice through fast heuristics based only on the structure of the network: they, indeed, assign scores to the nodes in the graph defining their ``importance'', and then they simply return the nodes with the highest scores.

In this work, we propose to extend this approach in order to encompass the problem of election manipulation. Specifically, this work introduces several heuristics in that sense: they are based on both the structure of the networks and centrality metrics adapted to the election context.
We next present in details these heuristics.

The core idea involves computing the distance on the political spectrum of a voter from the target candidate. Among the supporters of candidates other than the target $c^*$, the closer is this voter to $c^*$, the higher is the importance of this voter for manipulation purposes.
However, the score must also take into account social influence and the ability of a (possibly useless for manipulation purposes) voter to influence other (useful for manipulation purposes) voters. These ideas lead to two different classes of heuristics that we will describe below, and whose performances are described in the next section.

\paragraph{Scoring-based Heuristics.}
The first proposed class of heuristics considers the nodes in the neighbourhood $N(v)$ of each voter $v$ and their inclination to support the target candidate. Different heuristics in this class can be distinguished by how the neighborhood $N(v)$ is defined. For example,
\begin{itemize}
    \item we can consider in $N(v)$ only neighbours up to a given distance. For instance, setting this limit to 1, we consider only $v$'s friends; setting it to 2, we consider $v$'s friends and, in turn, their friends, too.
    \item $N(v)$ can be limited by a maximum number of nodes to consider, regardless of the distance. One can consider the neighbours at increasing distances from $v$, as long as the number of neighbours does not exceed a given threshold.
    \item we can limit $N(v)$ by both the constraints above.
\end{itemize}

Given $N(v)$, the higher the number of friends of $v$, the higher is the probability of exerting social influence. Clearly, this probability also depends on the distance from $v$, since neighbors of $v$ only require that influence from $v$ is successful, while for a friend of a friend $u$ of $v$ we require that not only the influence of $v$ but also the influence of $u$ is successful. These considerations justify the definition of a first score $s_G(v) = \sum_{u \in N(v)} 1/d(v, u)$ for the node $v$ purely based on the structure of the network $G$,
% \begin{equation}
% % \label{eq:s_G}
%     s_G(v) = \sum_{u \in N(v)} 1/d(v, u)
% \end{equation}
where $d(v, u)$ is the length of the shortest path linking $v$ to $u$.

We also consider political information about voters in the neighbourhood of $v$: indeed, we define
% another metric
$s_P(v) = \sum_{u \in N(v) \cup \{v\} : F[ d_P(u, c^*) ] = 1} w(v, u) \cdot 1/d_P(u, c^*)$ for the node $v$,
% as follows:
% \begin{equation}
% \label{eq:s_P}
%     s_P(v) = \sum_{u \in N(v) \cup \{v\} : F[ d_P(u, c^*) ] = 1} w(v, u) \cdot 1/d_P(u, c^*)
% \end{equation}
where $d_P(u, c^*)$ is a political distance function, that measures the distance on the political spectrum between $c^*$ and voter $u$; $F[\cdot]$ is a custom filter function, allowing to discard nodes depending on the distance $d_P$ (e.g., we may discard nodes $u$ such that $d_P(u, c^*)=0$, which means that they already vote for $c^*$); $w(v, u)$ is the product of the probabilities of the edges along the shortest path from $v$ to $u$. In this way, nodes that are more likely to be influenced by $v$ are assigned a high weight. We assume that $w(v, v)=1$.
Hence, $s_P(v)$ is high when the neighbourhood of $v$ is large, neighbours' political positions are not too far from $c^*$'s, and $v$ is likely to influence his neighbourhood, as desired.

We can then achieve different heuristics by combining these two scores in different ways. The following are possible combinations:
\begin{itemize}
    \item $s_1(v) = (s_G(v), s_P(v))$
    \item $s_2(v) = (s_P(v), s_G(v))$
    \item $s_3(v) = \alpha \hat{s}_P(v) + (1-\alpha) \hat{s}_G(v)$
\end{itemize}
where $\hat{s}_G(v)$ is $s_G(v)$ divided by the standard deviation of $s_G(u)$ over all the nodes $u \in V$; similarly for $\hat{s}_P(v)$. By using $s(v)=(X, Y)$, we mean that when comparing the scores of two nodes $u$ and $v$, the comparison is performed on the value of $X$, whereas $Y$ breaks ties. Having defined the score of the nodes, the algorithm greedily chooses the $B$ nodes $v$ with the highest score $s(v)$.

\paragraph{PageRank-based Heuristics.}
Motivated by the fact that PageRank is a fast heuristic that embodies the structure of the whole network, this work also suggests a second fast heuristic based on a special version of PageRank specifically thought for the problem of election manipulation. In particular, we observe that while the original PageRank shares the importance of a node among all its neighbors with probability $1-s$, and among all nodes with probability $s$, this does not makes really sense in our setting, since we already known that there are nodes that are more useful to the goal (i.e., the ones that may be lead to vote for $c^*$), and others that are less useful to this purpose (i.e., the ones that already vote for the desired candidate or the ones that cannot be made to vote for $c^*$).

Our idea is to share the rank of node $u$, by first assigning to each neighbor $v$ a weight $z_u(v) \geq 0$, quantifying how strongly $v$ positively influences other nodes to switch to $c^*$, and to share the PageRank of $u$ among its neighbors proportionally to their weight.
It is not hard to check that this modification does not affect the Markovian nature of PageRank computation, and hence the fact that it eventually converges to a stable set of values.

How should the weights be selected? Scores defined above already provide a measure of how much the social influence of the node serves to the purpose of the manipulator. Hence, we can for example set $z_u(v) = s_3(v)$ for a specific value of $\alpha$. Another variant of this heuristics, can be achieved by including in the weight $z_u(v)$ also information related to the diffusion probabilities, i.e., $z_u(v) = s_3(v) \cdot p(u, v)$.
In this way, this heuristics tends to assign a high score to nodes:
\begin{itemize}
    \item generally important for the diffusion in the network (inheriting PageRank properties about influence maximization), and
    \item with a high probability of influencing their neighbourhood, and
    \item whose neighbourhood does not already support $c^*$ (need to choose an appropriate filter function $F[\cdot]$), and
    \item whose neighbourhood contains nodes that are not too far from the target candidate on the political spectrum (thus, they are willing to change their mind in favour of $c^*$).
\end{itemize}
% So, the manipulator can greedily choose the $B$ nodes with the highest score.

\paragraph{The Full List of Heuristics.}
Based on the above ideas, we next describe the heuristics that we considered in this work.

All heuristics adopting the neighbourhood score use the following distance function:
\[
d_P(u, c^*) = \begin{cases}
  0 & \text{if } u \text{ votes for $c^*$ before the manipulation}\\
  |x_u - x_{c^*}| & \text{otherwise}
\end{cases}
\]
This highlights the reasons for which this heuristics should perform worse in nearly-single-peaked scenarios: when $\NOISE \neq 0$, $x_{c^*} \neq x^v_{c^*}$ and the value of $d_P(u, c^*)$ does not reflect the distance from $c^*$ perceived by the voter $u$ (and the manipulator could also guess that a voter votes for $c^*$ when he does not actually do).

In addition, all the versions of the algorithms adopting the neighbourhood score use the following filter function:
\[
F[a] = \begin{cases}
  1 & \text{if } a>0 \\
  0 & \text{otherwise}
\end{cases}
\]
This means that when computing the political scores $s_P(v)$, the algorithms discard all the voters that already vote for $c^*$ in the neighbourhood $N(v)$.

The neighbourhood $N(u)$ was limited to consider (1) friends of node $u$ or (2) friends of $u$ and friends of $u$'s friends. No limit was applied to the size of the neighbourhood. Moreover, the variants of the algorithms consider three ways of combining scores $s_P$ and $s_G$ to obtain the score $s$, namely $s=(s_G, s_P)$, $s=(s_P, s_G)$, $s=\alpha \hat{s}_P + (1-\alpha)\hat{s}_G$.

In conclusion, the algorithms based on the neighbourhood heuristic are:
\begin{itemize}
    \item \emph{SPoutdeg}. The neighbourhood of a node only considers its friends; the combined score is $s=(s_G, s_P)$.
    \item \emph{SPoutdeg\_rev}. The neighbourhood of a node only considers its friends; the combined score is $s=(s_P, s_G)$.
    \item \emph{SPoutdeg\_merge0.5}. The neighbourhood of a node only considers its friends; the combined score is $s=0.5\cdot \hat{s}_P + 0.5\cdot \hat{s}_G$.
    \item \emph{SPneig2}. The neighbourhood of a node considers its friends and friends of friends (2 is the maximum number of hops to reach the nodes in the neighbourhood); the combined score is $s=(s_G, s_P)$.
    \item \emph{SPneig2\_rev}. The neighbourhood of a node considers its friends and friends of friends; the combined score is $s=(s_P, s_G)$.
    \item \emph{SPneig2\_merge0.5}. The neighbourhood of a node considers its friends and friends of friends; the combined score is $s=0.5\cdot \hat{s}_P + 0.5\cdot \hat{s}_G$.
\end{itemize}

We can now present the variants of the weighted PageRank algorithm that use the neighbourhood score when sharing PageRank among nodes. These are the PageRank heuristics that use the previously defined distance function $d_P$ and the previously defined filter function $F$:
\begin{itemize}
    \item \emph{SPpagerank1.0\_pos}. The scores $s$ use the neighbourhoods considering only the friends of the nodes. In particular, $s=0\cdot \hat{s}_G + 1\cdot \hat{s}_P=\hat{s}_P$. \emph{pos} stands for \emph{positive filter} $F$.
    \item \emph{SPpagerank0.5\_pos}. The scores $s$ use the neighbourhoods considering only the friends of the nodes. In particular, $s=0.5\cdot \hat{s}_G + 0.5\cdot \hat{s}_P$. \emph{pos} stands for \emph{positive filter} $F$.
    \item \emph{SPpagerank1.0\_hop2\_pos}.  The scores $s$ use the neighbourhoods considering the friends of the nodes and the friends of friends. In particular, $s=0\cdot \hat{s}_G + 1\cdot \hat{s}_P$. \emph{pos} stands for \emph{positive filter} $F$. \emph{hop2} stands for the maximum number of hops to build the neighbourhood.
\end{itemize}

In addition, there are other two algorithms based on PageRank. Recall that $M$ has been defined as the set of voters that, if influenced, will vote for $c^*$. Then we have the following algorithms:
\begin{itemize}
% in code, this is "SPpagerank1.0_stepdist_eq1"
    \item \emph{SPpagerank1.0\_manip\_eq1}. The scores $s$ use the neighbourhoods considering only the friends of the nodes. In particular, $s=\hat{s}_P$. The distance function assigns $1$ to manipulable nodes (namely the ones in $M$), $0$ to voters that already vote for $c^*$, and $\infty$ to the rest of the voters (this represents the fact that nodes are not manipulable in favour of $c^*$). The filter function is
    \[
        F[a] = \begin{cases}
                1 & \text{if } a=1 \\
                0 & \text{otherwise}
        \end{cases}
    \]
    In the friendly name, \emph{manip} stands for \emph{manipulable} and \emph{eq1} is a short name for the filter function.
% in code, this is SPpagerank1.0_stepdistGAIN_pos
    \item \emph{SPpagerank1.0\_manip*\_pos}. The scores $s$ use the neighbourhoods considering only the friends of the nodes. In particular, $s=\hat{s}_P$. The distance function assigns $1$ to manipulable nodes (namely the ones in $M$), $0$ to voters that already vote for $c^*$, and
    \[
    \frac{|x_v-x_{c^*}|}{|x_v-x_{c^*}| - |\hat{x}_v-x_{c^*}|}
    \]
    to the rest of the voters $v$, where $\hat{x}_v$ is the position of voter $v$ if he was influenced. \emph{manip*} stands for the fact that the algorithm works as \emph{SPpagerank1.0\_manip\_eq1} does but considers the gain (in getting closer to $x_{c^*}$) of non-manipulable voters.
\end{itemize}
% We here only present in detail the one that happens to achieve the best experimental performances, namely \emph{SPpagerank1.0\_pos} (we refer to Appendix~\ref{apx:heuristics} for a description of the other heuristics).
% %
% The idea of this heuristics is to share the rank of node $u$ not uniformly to all its neighbors (as in the standard PageRank), but proportionally to a score $\hat{s}$ assigned to each node. The score $\hat{s}$ of $v$ considers political information about voters in the neighbourhood of $v$: specifically we define
% $s(v) = \sum_{u \cup \{v\} \colon (v, u)\in E \text{ and } d(u, c^*) > 0} p(v, u) \cdot 1/d(u, c^*)$ for the node $v$,
% where $d(u, c^*)$ measures the distance on the political spectrum between $c^*$ and voter $u$. In this way, nodes that are more likely to be influenced in favour of $c^*$ by $v$ are assigned a higher weight. Score $\hat{s}$ is then achieved from $s$ by simply normalizing the latter function, by dividing its value by the standard deviation over all nodes.

\paragraph{Complexity of Heuristics.}
% \label{apx:complexity}
To analyze the complexity of the neighbourhood heuristics, Algorithm \ref{alg:neighbourhood_complexity_sketch} sketches the main required steps.
%%%%%%%%%%%%%%%%% neighbourhood heuristics %%%%%%%%%%%%%%%%
\begin{algorithm}
\caption{Neighbourhood heuristic}\label{alg:neighbourhood_complexity_sketch}
\begin{algorithmic}
\Require the graph of voters $G=(V, E)$, the target candidate $c^*$

\For {$v$ in $V$}
    \LineComment{Perform a BFS traversal not violating the constraints on the limit of the neighbourhood}
    \State $N(v) = BFS(v)$

    \State $s_G(v) = \sum_{u \in N(v)} 1/ d(v, u)$
    \LineComment{$\beta$ is the weight function combining $F$ and $w$}
    \State $s_P(v) = \sum_{u \in N(v)} \beta(u,v)/d_P(u, c^*)$
\EndFor

\LineComment{Standardize the scores for each $v \in V$}
\State $\hat{s}_G(v) = s_G(v) / STD(s_G(\cdot))$
\State $\hat{s}_P(v) = s_P(v) / STD(s_P(\cdot))$

\LineComment{Combine $\hat{s}_G(v)$ and $\hat{s}_P(v)$ for each $v \in V$}
\State $s(v) = ...$ \Comment{Depends on the specific score}

\State \textbf{return} $s$
\end{algorithmic}
\end{algorithm}
%%%%%%%%%%%%%%%%%%%%%%%%%%%%%%%%%%%%%%%%%%%%%%%%%%%%%%%%%%
We assume that the function $d_P(u, c^*)$ finds the most preferred candidate of voter $u$. In this way, $d_P(u, c^*)$ can be $0$ if $u$ votes for $c^*$, even if the distance between $u$ and $c^*$ is not $0$. Hence the complexity of $d_P(u, c^*)$ is $O(m)$, where $m$ is the number of candidates. As explained above, most of the experiments are performed considering that $N(v)$ contains only the friends of $v$; hence $s_G(v)=outdeg(v)$ and can be computed in $O(1)$ time with an appropriate representation of the graph. Similarly, $s_P(v)$ can be computed in $O(outdeg(v)\cdot m)$. The total cost to compute $s_G$ and $s_P$ is
\begin{equation}
    \sum_{v \in V} O(outdeg(v) \cdot m)  = O(|E|\cdot m)\nonumber
\end{equation}
The standardization costs $O(|V|)$; hence, the total cost to compute the neighbourhood heuristic is $O(|V| + m \cdot |E|)$.

Compared to the standard algorithm, the weighted version of the PageRank algorithm does not require any substantial extra cost. Thus, for simplicity, the calculations of the computational complexity to compute the PageRank heuristic are based on the original PageRank algorithm. Algorithm \ref{alg:pagerank_sketch} sketches the main steps required to calculate PageRank values of the nodes of a given graph.
%%%%%%%%%%%%%%%%% pagerank algorithm %%%%%%%%%%%%%%%%
\begin{algorithm}
\caption{PageRank - computational complexity: $O(k\cdot(|E|+|V|))$}\label{alg:pagerank_sketch}
\begin{algorithmic}
\Require the graph of nodes $G=(V, E)$, the number of iterations $k$, the scaling factor $s$

\State $r(v)=1/|V|, \forall v \in V$

\For {$i = 1,...,k$}

\State $r^*(v)=0, \forall v \in V$

\For {$v$ in $V$}
    \LineComment{Apply the update rule}
    \For {$u$ in out-neighbours of $v$}
        \State $r^*(u) \mathrel{+}= r(v)/outdeg(v)\cdot s$
    \EndFor
    \State $r^*(v) \mathrel{+}= (1-s)/|V|$

\EndFor
\State $r(v) = r^*(v), \forall v \in V$
\EndFor

\State \textbf{return} $r$
\end{algorithmic}
\end{algorithm}
%%%%%%%%%%%%%%%%%%%%%%%%%%%%%%%%%%%%%%%%%%%%%%%%%%%%%%%%%%
% It is a simple implementation of the update rule defined in section \ref{sec:pagerank}.
The update rule for all the nodes requires
\begin{equation}
    \sum_{v \in V} \sum_{u \in N^+_v} O(1) = \sum_{v \in V} O(outdeg(v)) = O(|E|) \nonumber
\end{equation}
where $N^+_v$ is the set of out-neighbours of node $v$. At each iteration, an additional cost $O(|V|)$ is required to manage $r$ and $r^*$. Since the update rule is applied $k$ times to converge to stable values, the total computational complexity is $O(k\cdot(|E|+|V|))$. Note that this is a naive implementation of the PageRank algorithm. It is presented here only to pose an upper bound to the computational complexity of the PageRank heuristics. In fact, solving the problem by finding a particular eigenvector of a matrix related to the graph can speed up the computations. In fact, the implementation used for the experimental phase relies on \emph{scipy}'s methods to compute eigenvectors, which are based on ARPACK, a library specifically designed for solving large scale eigenvalue problems. Hence, the actual performances of the algorithm can be much better than implementing Algorithm \ref{alg:pagerank_sketch}. Moreover, optimized implementations of PageRank can surely scale up to large graphs.

\section{Experimental Results}
\paragraph{Experimental Setting.}
We run extensive experiments to evaluate the algorithms described above both in terms of the effectiveness of the manipulation, as measured by the margin of victory $MoV$ and the change of margin of victory $\Delta MoV$ resulting from the simulated manipulations,
% as defined in equation \ref{eq:DeltaMoV},
and the execution time of the algorithms. As for the valuation of the running time,  we used the following software and hardware equipment:
\begin{itemize}
    \item \emph{Windows} 10 (listed here for reproducibility of the execution times);
    \item \emph{Python} 3.8.10;
    \item \emph{scipy}, version 1.8.1 (listed here for reproducibility of the PageRank algorithm).
\end{itemize}
The details of the hardware environment are:
\begin{itemize}
    \item RAM: 8 GB;
    \item CPU: i5-6400, 2.70 GHz.
\end{itemize}
We stress that we compare \emph{Python} implementations run on a single core, without  any code optimization.
% with compiled programming languages, specific compile options, parallel architectures and similar features.
% Indeed, we only care about how execution times grow with larger and larger inputs to analyze the scalability of the algorithms and how they would perform in real-world settings. Moreover, this does not prevent us from making comparisons between the running times of different algorithms, as long as they are simulated on the same machine.

% The algorithms were tested on several instances of the election manipulation problem. The performances were measured in terms of average $\Delta MoV$. Moreover, in the following, the standard deviation of the results is analyzed, too.
% ; this can help choose the best algorithm among the tested ones. In fact, even if an algorithm has a slightly lower average performance than another one, it could show lower standard deviations proving that it is more reliable.

% In general, the size of the problems (election and graph) of the simulations was limited by computational constraints. For the same reason, not all the algorithms were tested with all the combinations of parameters presented in the following.
%
% ELECTION PARAMETERS
% Here is the list of the values used to create the
In general we compared different algorithms in election scenarios with a set of voters
% \begin{description}
%     \item[number of voters:] the algorithms were tested on electorates
    of sizes 20, 50, and 100.
%     A particular experiment to test the scalability of some of our heuristics have been run over a larger electorate. Moreover a
    All experiments
%     always
    involved five candidates.
%     , which is a fair value for an election.
%     The experiments did not analyze results with more candidates, mainly for computational reasons. More details about this aspect are explained at the end of this chapter.% (non si può andare più di 5!)
% \end{description}
% ELECTION PT.2 (voters, positions, candidates positions, noise)
% The details of the elections involve the positions of the voters and the candidates on the political spectrum.
Both voters and candidates were assigned positions in the range $[-1, +1]$ chosen uniformly at random. This should create scenarios in which all the parties are almost equally supported by the voters. When dealing with nearly-single-peaked electorates, the noise was chosen to be independent of the positions of the candidates, namely $\NOISE(x_c)=\NOISE$. Specifically, we considered
% following parameters were tested:
% \begin{description}
%     \item[
    \emph{uniform noise} $\NOISE=U(-0.2;0.2)$;
%     , namely a uniform random variable;
%     \item[
    \emph{Gaussian noise with low variance} $\NOISE = \mathcal{N}(0;0.08)$;
%     \item[
    and \emph{Gaussian noise with high variance} $\NOISE = \mathcal{N}(0;1)$.
%     Note that, compared to the size of the political spectrum, the variance is really big.
% \end{description}
These values were chosen to test different levels of nearly single-peaked average swap distances. In fact, when normalized with respect to the number of voters, for an election with 20 voters and 5 candidates randomly placed on the political spectrum, the average distances from single-peaked scenarios are $\sim 0.5$, $\sim 1$, and $\sim 2$ swaps (see also Figure \ref{fig:kswaps_per_noise_distribution}).
\begin{figure}[H]
    \centering
    \includegraphics[width=1.0\textwidth]{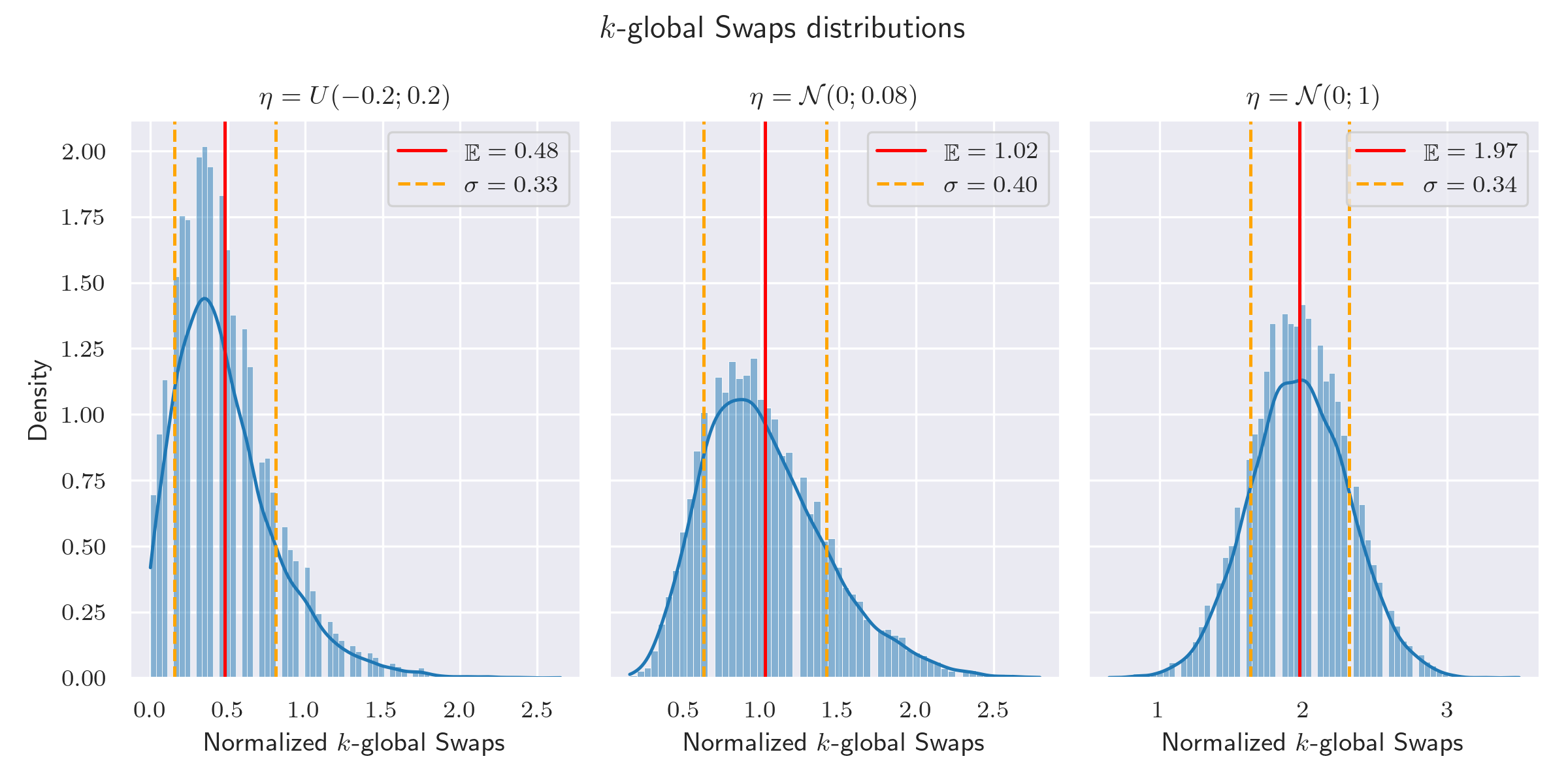}
    \caption{Distributions of $k$-global swaps distances (normalized with respect to 20 voters) for three different noises. Data of the distributions derive from 10000 elections involving 20 voters and 5 candidates placed on the political spectrum uniformly at random. The plots show the average $k$-swap distance (red) and the deviation from the mean (orange).}
    \label{fig:kswaps_per_noise_distribution}
\end{figure}

As for the parameters describing the power of the manipulator, we consider
%
% are the ones presented above, namely the budget and the number of manipulation campaigns. Moreover, the parameter $\delta$ in \eqref{eq:voter_manipulation} contributes to the effectiveness of the manipulation. For the experimental validation, these are the values that were used:
% \begin{description}
%     \item[
    budget $B$ such that
%     :] to make results about different sizes of the electorate comparable, the budget is expressed as the percentage of voters the manipulator can bribe. In particular, the manipulations were simulated with
    $\frac{B}{|V|} \in \{5\%, 10\%, 15\%\}$;
%     \item[
    the maximum number of manipulation campaigns
%     ]
    was set to 10;
%     Real-world election campaigns are supposed to last no more than one month. Therefore, since each manipulation round of the proposed model starts with feedback about the results of the previous campaign, it is reasonable to assume that the manipulator cannot repeat this process more than a dozen times.
%     \item[
    the parameter $\delta$ that represents the influenceability of the electorate
%     or its education
    has been chosen in
%     . The tests are performed with $\delta \in
    $\{0.1, 0.2, 0.3, 0.4\}$;
%     to compare the results of different education levels of the voters.
% \end{description}
the target candidate of the manipulation process was chosen randomly among all the candidates $\{c_0, c_1, c_2, c_3, c_4\}$ unless stated differently.

Monte-Carlo simulations necessary for the approximation algorithm have been repeated 300 times (corresponding to $\lambda = \sqrt\frac{0.1}{2}$ and $\epsilon = 0.05\cdot\sqrt 2$ in Proposition~\ref{proposition_weightedinfluencemaximizationzeroone}). Moreover, estimates of $\Delta MoV$ have been evaluated over a number of simulation sufficient to achieve statistical guarantee. Specifically, the number of simulation have been selected according to the following results.
\begin{proposition}%[Bounded variables, additive approximation]
\label{proposition_boundedvariablesadditiveapproximation}
Let $X_1, \ldots, X_T$ be i.i.d random variables such that $a \leq X_i \leq b$ with probability 1, $\forall i \in \{1,...,T\}$. Let $\mu = E[X_1]=...=E[X_T]$.

The average $\hat\mu = 1/T\cdot \sum_{i=1}^T X_i$ is an approximation to $\mu$ with an error smaller than $\epsilon$, with probability at least $1-2\lambda^2$, if

\begin{equation}
T \geq \left(\frac{b-a}{\epsilon}\right)^2 ln\left( \frac{1}{\lambda}  \right) \nonumber
\end{equation}
\end{proposition}
\begin{proof}
Let $X=\sum_{i=1}^T X_i$.
Applying Hoeffding's inequality we get:

\begin{equation}
\label{eq_hoeffdinginequality_propboundedvarsaddapx}
    Pr\left(  \abs{X-E[X]} \geq \gamma  \right) \leq 2 \exp \left(         -\frac{2\gamma^2}{\sum_{i=1}^T {(b-a)^2}}  \right) \forall \gamma>0
\end{equation}

Rearranging the left side of \eqref{eq_hoeffdinginequality_propboundedvarsaddapx} we get
\begin{equation}
    Pr\left( \abs{\frac{X}{T} -\mu } \geq \frac{\gamma}{T} \right) = Pr\left( \abs{\hat\mu -\mu } \geq \frac{\gamma}{T} \right)\nonumber
\end{equation}
Choosing $\gamma = T\epsilon$ and rearranging the right side of \eqref{eq_hoeffdinginequality_propboundedvarsaddapx} we get
\begin{align*}
     Pr\left( \abs{\hat\mu -\mu } \geq \epsilon \right) \leq 2 exp \left( -\frac{2\gamma^2}{\sum_{i=1}^T {(b-a)^2}}  \right) &=  2 exp \left(-\frac{2}{T(b-a)^2} T^2 \epsilon^2 \right) \tag*{by substituting $\gamma$}\\
    &=  2 exp \left(-\frac{2}{(b-a)^2} T \epsilon^2 \right)\\
    &\leq 2exp\left( -2ln\left(\frac{1}{\lambda}\right)  \right)\tag*{by substituting $T$}  \\
    &= 2\lambda^2\tag*{\qed}
\end{align*}
\end{proof}

\begin{corollary}
\label{cor:simulDmov}
 With $\left(\frac{|V| - |V_{c^*}| + |V_{\bar c}|}{\epsilon}\right)^2 ln\left( \frac{1}{\lambda}  \right)$ simulations, where $|V_{c^*}| $ is the number of votes for $c^*$ before the manipulation and $|V_{\bar c}|$ is the number of votes for the best opponent of $c^*$ before the manipulation, we can evaluate the expected $\Delta MoV$ with an error smaller than $\epsilon$, with probability at least $1-2\lambda^2$
\end{corollary}
\begin{proof}
Proposition \ref{proposition_boundedvariablesadditiveapproximation} can be used to compute the number of simulations needed to approximate the statistical mean of $\Delta MoV$. The definition of the increment of the margin of victory of the target candidate $c^*$ is reported here for convenience:
% (the same as equation \ref{eq:DeltaMoV}):
\begin{equation}
\label{eq_MoV_definition}
    \Delta MoV = |V_{c^*}^*| - |V_{\hat c}^*| - \left( |V_{c^*}| - |V_{\bar c}| \right)
\end{equation}
where
\begin{itemize}
\item $|V_{c^*}^*|$ is the number of votes for $c^*$ after the manipulation;
\item $|V_{\hat c}^*|$ is the number of votes for the best opponent of $c^*$ after the manipulation;
\item $|V_{c^*}| $ is the number of votes for $c^*$ before the manipulation; it does not depend on the manipulation algorithm;
\item $|V_{\bar c}|$ is the number of votes for the best opponent of $c^*$ before the manipulation; it does not depend on the manipulation algorithm.
\end{itemize}

The worst result of the manipulation algorithm is that the MoV does not change at all since supporters of $c^*$ cannot negatively change their idea. The best possible result of the manipulation algorithm is $|V_{c^*}^*|=|V|$ and $|V_{\hat c}^*|=0$, i.e., the whole electorate supports $c^*$. Hence, in a generic simulation of the manipulation, $\Delta MoV$ is bounded

\begin{equation}
    0 \leq \Delta MoV(\cdot) \leq |V| - |V_{c^*}| + |V_{\bar c}| \nonumber
\end{equation}

So, $E[\Delta MoV(\cdot)]$ can be estimated by averaging the results of repeated simulations, and the number of such simulations can be computed by using Proposition \ref{proposition_boundedvariablesadditiveapproximation} with $a=0$ and $b=|V| - |V_{c^*}| + |V_{\bar c}|$, for some fixed values of $\epsilon$ and $\lambda$.\qed
\end{proof}
Hence, the number of simulation run for estimating $\Delta MoV$ has been set as suggested by Corollary~\ref{cor:simulDmov} with $\lambda = \sqrt\frac{0.1}{2}$ and $\varepsilon = 0.05$.
% As for the parameters of the tools we used, namely the greedy algorithms for influence maximization and the page rank algorithm, we have done the following choices:
% \begin{itemize}
%  \item the parameter $\lambda$ describing the error probability has been set to $\sqrt\frac{0.1}{2}$;
%  \item the parameter $\epsilon$ describing the error range has been set to $0.05\cdot\sqrt 2$ for the Monte Carlo simulations required by the greedy algorithm for the influence maximization, that (according to Proposition~\ref{proposition_weightedinfluencemaximizationzeroone} in Appendix~\ref{apx:simulations}) implies a number of 300 Monte Carlo simulations for estimating each variable;
%  \item the parameter $\epsilon$ describing the error range on estimating $\Delta MoV$ has been further lowered to $0.05$.
% \end{itemize}

All the variants of the models and algorithms described above have been tested against both synthetic and real-world networks.
% synthetic networks
Specifically, we used:
% \begin{description}
%  \item[
(i) Watts-Strogatz graphs \cite{watts1998collective}
%  ]
%  with the following parameters:
% \begin{itemize}
%     \item
with
    nodes uniformly distributed in the 2D square whose side is $\sqrt \frac{|V|}{20}$ (thus, the density of the nodes remains the same increasing the size of the electorate);
%     \item
    strong ties with a radius $r=0.13$;
%     : a node is linked to nodes at distance at most $r$ (strong ties);
%     \item
    $k=2$ weak ties
%     : the number of weak ties.
distributed inversely to the distance with a power law of exponent
%     \item
    $q=2$.
%     : the exponent of the probability distribution to create weak ties.
% \end{itemize}
Since the density does not change when the number of nodes increases, the degree of each node remains almost the same.
% \item[
(ii) Preferential attachment graphs \cite{barabasi1999emergence}, created by setting the probability of linking preferentially to $0.25$ and $0.75$.
% \end{description}
% Note that Watts-Strogatz models could represent real societies in which voters are linked by geographical proximity, while preferential attachment graphs could represent online social networks, such as Twitter, where a few people get global visibility, being followed by a significant fraction of the population, and many accounts only achieve negligible influence in the whole graph. The probabilities of preferential attachment were chosen to compare results in settings showing more or less accentuated rich-get-richer effects.

For these networks, tests were performed on several combinations of parameters. Specifically, for each pair $(|V|, B/|V|) \in \{20, 50, 100\} \times \{0.05, 0.10, 0.15\}$, the simulated scenarios involved (in all the possible combinations):
% \begin{itemize}
%     \item
    8 random placements of voters and candidates on the political spectrum;
%     \item
    10 randomly generated graphs (Watts-Strogatz or preferential attachment models);
%     \item
     10 randomly generated sets of diffusion probabilities on edges.
% Probabilities were chosen uniformly at random in $[0,1]$.
% \end{itemize}
This led to $8 \times 10 \times 10 = 800$ electoral scenarios.
% In nearly-single-peaked elections, the blurred views of the voters change randomly according to the graph, i.e., each of the 10 generated graphs is related to a random noisy view of the political positions. To avoid any correlation between the generated data, the three random elements above are generated independently. Recall that for each electoral scenario we run multiple simulations in order to achieve a close approximation of $\Delta MoV$ (as given by Corollary~\ref{cor:simulDmov} in Appendix~\ref{apx:simulations} according to the choice of $\lambda$ and $\epsilon$ described above).
%
Due to the large running time of our approximation algorithm, it has been tested only with:
% \begin{itemize}
%     \item
    31 random placements of voters and candidates;
%     on the political spectrum;
%     \item
    5 randomly generated graphs (Watts-Strogatz models only);
%     \item
    5 randomly generated sets of diffusion probabilities.
%     of the edges of the graph.
%     Probabilities were chosen uniformly at random in $[0,1]$.
% \end{itemize}
This led to 775 random elections.

% real network
The real case study involves a snapshot of the Facebook social network \cite{leskovec2012learning} available at SNAP \cite{leskovec2016snap}.
The network consists of 4039 nodes and 88234 undirected edges.
% We describe the details about how it has been adapted to the election manipulation purposes in Appendix~\ref{apx:facebook}.
It was adapted to the purposes of this work by associating each node with a random position on the political spectrum and assigning a random diffusion probability to each edge. However, generated data were not totally random: the testing phase considered the structure of the graph to create plausible data. The net was first partitioned using the Louvain Method, a greedy algorithm to detect the communities the net consists of. Other approaches were tested, too; however, the Louvain algorithm returned the best partition. By measuring some standard quality metrics, the Louvain partitions obtained:
\begin{itemize}
    \item modularity 0.83. It lies in $[-1/2, +1]$ and represents the concentration of edges within the communities compared to random distributions of the edges; the higher the value, the better the partition;
    \item coverage 0.96; it is the ratio of the number of intra-community edges to the total number of edges of the graph;
    \item performance 0.92; it is the number of intra-community edges plus inter-community non-edges divided by the total number of potential edges. The higher the value, the better the partition.
\end{itemize}
The algorithm returned 16 communities (the distribution of the nodes is shown in Table \ref{tab:louvain_communitites}; Figure \ref{fig:facebook_louvain_communities} also reports a pictorial representation of the communities).
\begin{table}[H]
\centering
\begin{tabular}{ll|ll}
Community ID & Number of nodes & Community ID & Number of nodes \\ \hline
0            & 19              & 8            & 237             \\
1            & 19              & 9            & 323             \\
2            & 25              & 10           & 350             \\
3            & 60              & 11           & 423             \\
4            & 73              & 12           & 432             \\
5            & 117             & 13           & 446             \\
6            & 206             & 14           & 535             \\
7            & 226             & 15           & 548
\end{tabular}
\caption{Sizes of the communities detected by the Louvain method on the Facebook network.}
\label{tab:louvain_communitites}
\end{table}
\begin{figure}[H]
    \centering
    \includegraphics[width=0.8\textwidth]{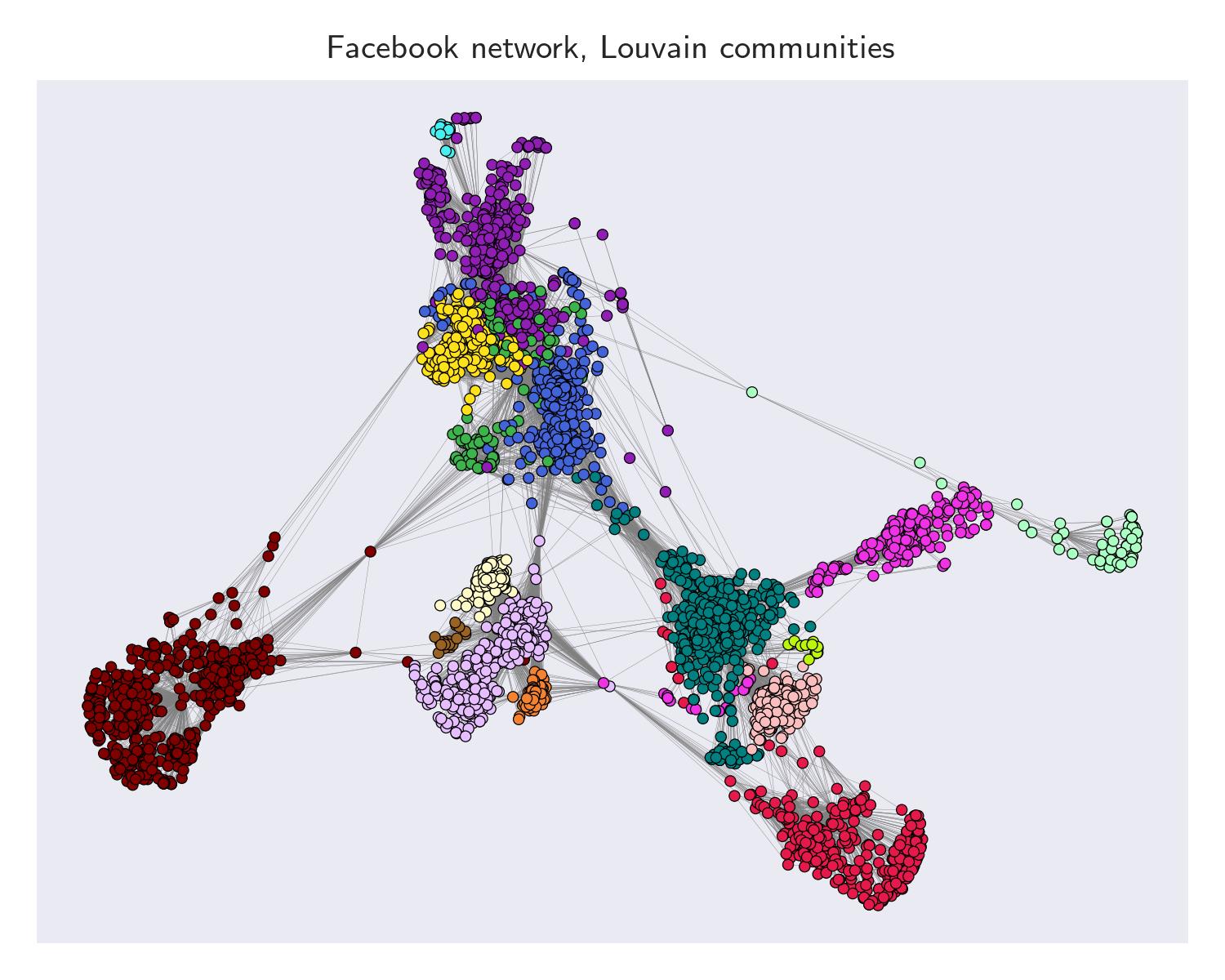}
    \caption{The Facebook network used for the experiments. Positions of the nodes were computed by using the Fruchterman-Reingold algorithm. The colours of the nodes represent the (Louvain) community they belong to.}
\label{fig:facebook_louvain_communities}
\end{figure}
Once the graph had been partitioned, communities were used to create reasonable diffusion probabilities of the edges. In particular, edges connecting nodes in the same group were assigned probabilities chosen uniformly at random in the range $[0.6, 1]$; edges connecting nodes in different communities were assigned probabilities chosen uniformly at random in the range $[0, 0.4]$. This should reflect the fact that people in the same circle of friends are more likely to share ideas, whereas nodes in distinct communities are less likely to do that.

The political parties of the experiment were equally spaced on the political spectrum. In particular, the candidates are (sorted from left to right) $c_0, c_1, c_2, c_3,$ and $c_4$, from $-1$ to $+1$. The target candidate is $c_2$. To test algorithm \emph{SPpagerank1.0\_pos} in the worst case for the target candidate, voters were placed on the spectrum as follows:
\begin{itemize}
    \item voters in the communities $0,1,2,3,4,5,6$ (i.e., the smallest ones) were placed uniformly at random in the range $[-0.25, +0.25]$. These voters vote for $c_2$.
    \item the other voters were placed uniformly at random in the ranges $[-1, -0.25]$ and $[0.25, 1]$. These voters do not vote for $c_2$.
\end{itemize}
Trying to impersonate the manipulator facing a real election manipulation problem, the experiment involves a single set of random diffusion probabilities and a single random electorate. Three noises were tested: $\NOISE=0$, $\NOISE=\mathcal{N}(0;1)$, and $\NOISE=\mathcal{N}(0;0.08)$. $\delta$ was set to $0.1$ and $0.3$; the budget was set to 5\% of the nodes. The test assumed that the manipulator was under the \emph{limited-knowledge} hypothesis. Since the graph is quite large, like the simulations of the approximation algorithm, the test was repeated 300 times, guaranteeing an error of $0.5\sqrt 2$ of the maximum $\Delta MoV$ with probability 90\% on the estimation of $\EXPECTED[\Delta MoV]$.

% stress test
Among all the tested algorithms, the best one (considering both manipulation performances and computational complexity) was tested in terms of scalability on very large graphs.
% For this purpose, we considered the elections involving up to 20.000 voters. The test consists of repeatedly performing the ten manipulation campaigns on several electorates without limiting the number of manipulations. The only limit is the simulation time; the experiments aim to observe the number of manipulations executed in 3 hours. If one can run the algorithm a significant number of times in a few hours, then we can conclude that the algorithm is not only effective at manipulating elections, but it is also fast enough to be exploited by an evil manipulator on problems concerning large societies.
%
% In order to run these scalability experiments, for each manipulation (that encompasses 10 manipulation campaigns), just one simulation is performed. This mimics the action of the manipulator: in fact, he cannot get the average $\Delta MoV$ on several manipulations of real societies, and he is limited to running a single manipulation.
The experiment is executed on Watts-Strogatz graphs built as explained above. The tested number of nodes of the graphs are $[200, 500, 1000, 2000, 5000, 10000, 20000]$. The algorithm is allowed to run for three hours for each graph size, repeatedly generating random graphs and elections. This experiment was repeated for $\NOISE=0$, $\NOISE=\mathcal{N}(0;0.08)$, $\NOISE=\mathcal{N}(0;1)$.

% %%%%%%%%%
%
% Finally, the number of manipulation campaigns the manipulator needs to run is notable, too. We prefer the algorithm that makes the target candidate win the election by exploiting the minimum number of manipulation rounds. This mimics the intention of the manipulator trying to save his budget.

\paragraph{Results.}
% Gli esperimenti che abbiamo li possiamo categorizzare più o meno come segue:
%
% -EXP1: confronto tra diverse euristiche e selezione di PageRank come migliore (e.g., Fig. 4.3)
% \paragraph{Comparison among heuristics.}
We start by comparing the different heuristics on Watts-Strogatz networks, in order to find the one that guarantees the best performances.
% Next we will show that similar performances are then achieved also on other networks.
%
% Specifically, we start by evaluating \emph{SPneig2\_merge0.5}, \emph{SPneig2}, \emph{SPneig2\_rev}, \emph{SPoutdeg\_merge0.5}, \emph{SPoutdeg}, \emph{SPoutdeg\_rev}, \emph{SPpagerank0.5\_pos}, \emph{SPpagerank1.0\_pos}, and \emph{SPpagerank1.0\_manip\_eq1}.
In Figure \ref{fig:results_wattsstrogatz_LK_STUDYING_STEP_2} we show the performances of a subset of some of the heuristics that we designed.
% the above algorithms: these are the bests among them; the worse ones are not shown to simplify the plots. Results are shown only for $\delta \in \{0.1, 0.3\}$. Similar results holds also for different values of this parameter.
%
% Interestingly, the best algorithms do not include \emph{SPneig2} and \emph{SPoutdeg} but do include \emph{SPneig2\_rev} and \emph{SPoutdeg\_rev}. Remember that the "reversed" versions use the political score $s_P$ to choose the best influencers and the $s_G$ to break ties. This means that the score $s_P$ is more effective at finding the right influencers than $s_G$; this makes sense since the manipulator is dealing with the problem of election manipulation (and not influence maximization) and, being $s_G$ based solely on the structure of the network, it is not as important as $s_P$.
%figure
\begin{figure}
% \vspace{-0.5cm}
    \centering
    \includegraphics[width=0.75\textwidth]{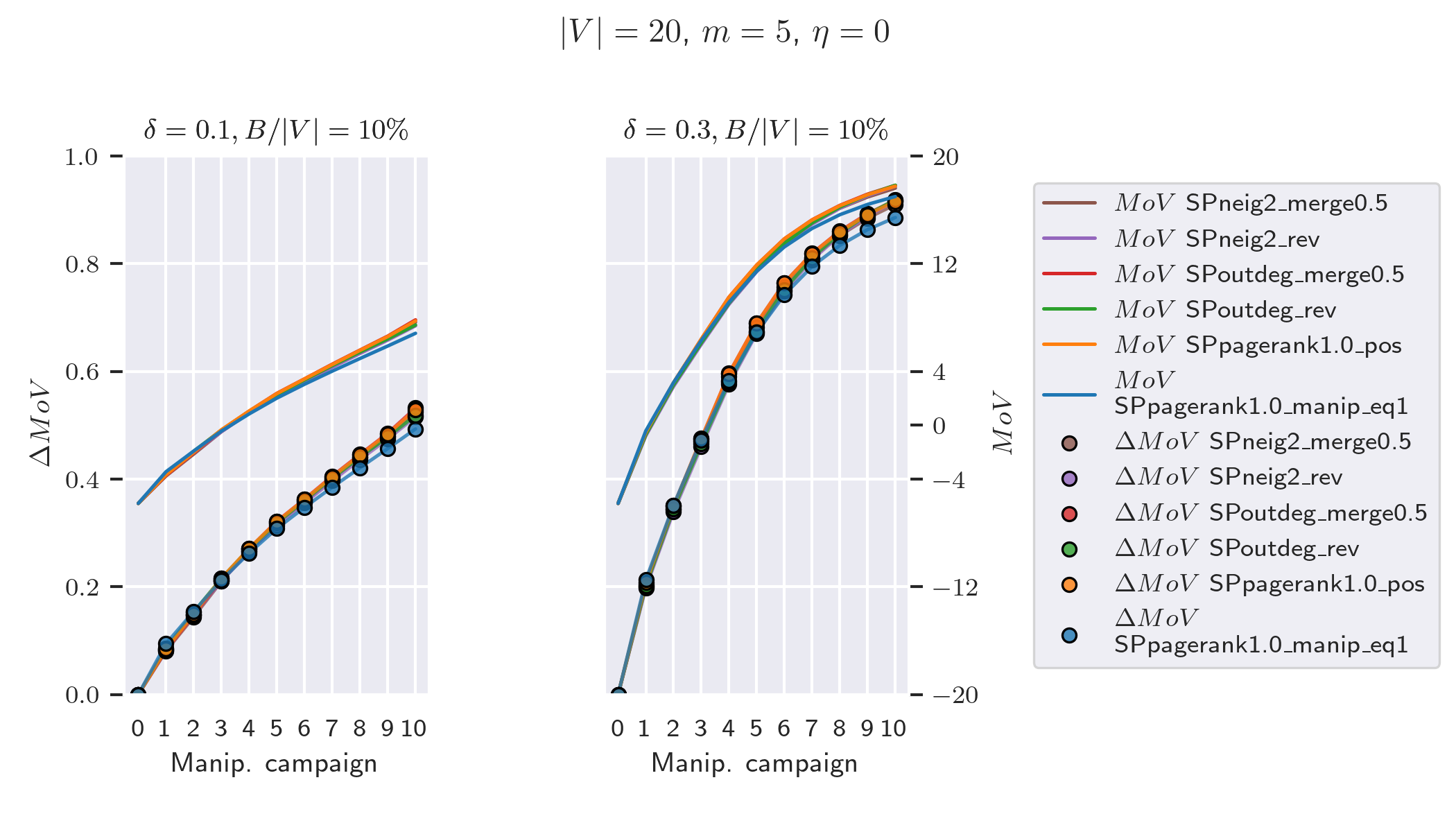}
    \caption{Performances of the algorithms.}
%     Each plot shows both the $\Delta MoV$ and the $MoV$.}
\label{fig:results_wattsstrogatz_LK_STUDYING_STEP_2}
% \vspace{-0.6cm}
\end{figure}
\begin{table}[htb]
\centering
\resizebox{\textwidth}{!}{
% this comes from pandas
\begin{tabular}{lllllllllllll}
\toprule
             &     &     &       Round 1 &       Round 2 &       Round 3 &       Round 4 &       Round 5 &       Round 6 &       Round 7 &       Round 8 &       Round 9 &      Round 10 \\
Alg & $\delta$ & B &               &               &               &               &               &               &               &               &               &               \\
\midrule
$PR1^+$ & 0.1 & 5\% &  0.074$\pm$0.053 &  0.133$\pm$0.068 &  0.190$\pm$0.090 &  0.240$\pm$0.112 &  0.286$\pm$0.127 &  0.325$\pm$0.135 &  0.362$\pm$0.142 &  0.399$\pm$0.149 &  0.434$\pm$0.157 &  0.472$\pm$0.162 \\
             &     & 10\% &  0.084$\pm$0.057 &  0.149$\pm$0.071 &  0.216$\pm$0.097 &  0.271$\pm$0.122 &  0.321$\pm$0.135 &  0.362$\pm$0.141 &  0.404$\pm$0.145 &  0.445$\pm$0.152 &  0.484$\pm$0.161 &  0.529$\pm$0.164 \\
             &     & 15\% &  0.091$\pm$0.059 &  0.161$\pm$0.073 &  0.233$\pm$0.104 &  0.291$\pm$0.130 &  0.344$\pm$0.143 &  0.387$\pm$0.144 &  0.433$\pm$0.148 &  0.476$\pm$0.156 &  0.519$\pm$0.167 &  0.570$\pm$0.168 \\
             & 0.3 & 5\% &  0.182$\pm$0.094 &  0.310$\pm$0.124 &  0.421$\pm$0.153 &  0.528$\pm$0.160 &  0.616$\pm$0.160 &  0.689$\pm$0.154 &  0.745$\pm$0.146 &  0.788$\pm$0.136 &  0.823$\pm$0.126 &  0.851$\pm$0.116 \\
             &     & 10\% &  0.208$\pm$0.102 &  0.351$\pm$0.128 &  0.475$\pm$0.156 &  0.596$\pm$0.153 &  0.690$\pm$0.146 &  0.765$\pm$0.132 &  0.819$\pm$0.119 &  0.860$\pm$0.107 &  0.891$\pm$0.095 &  0.916$\pm$0.085 \\
             &     & 15\% &  0.226$\pm$0.108 &  0.377$\pm$0.131 &  0.511$\pm$0.162 &  0.644$\pm$0.150 &  0.740$\pm$0.138 &  0.816$\pm$0.116 &  0.867$\pm$0.099 &  0.905$\pm$0.085 &  0.932$\pm$0.072 &  0.952$\pm$0.061 \\
$N2^{rev}$ & 0.1 & 5\% &  0.076$\pm$0.052 &  0.135$\pm$0.066 &  0.192$\pm$0.086 &  0.242$\pm$0.107 &  0.288$\pm$0.123 &  0.328$\pm$0.133 &  0.365$\pm$0.141 &  0.402$\pm$0.149 &  0.437$\pm$0.156 &  0.473$\pm$0.162 \\
             &     & 10\% &  0.083$\pm$0.056 &  0.148$\pm$0.070 &  0.211$\pm$0.092 &  0.266$\pm$0.116 &  0.314$\pm$0.132 &  0.356$\pm$0.141 &  0.397$\pm$0.148 &  0.437$\pm$0.156 &  0.475$\pm$0.163 &  0.516$\pm$0.168 \\
             &     & 15\% &  0.089$\pm$0.058 &  0.157$\pm$0.072 &  0.225$\pm$0.097 &  0.282$\pm$0.123 &  0.333$\pm$0.138 &  0.376$\pm$0.145 &  0.420$\pm$0.152 &  0.462$\pm$0.160 &  0.502$\pm$0.168 &  0.548$\pm$0.172 \\
             & 0.3 & 5\% &  0.181$\pm$0.092 &  0.309$\pm$0.122 &  0.419$\pm$0.149 &  0.523$\pm$0.160 &  0.612$\pm$0.161 &  0.687$\pm$0.156 &  0.748$\pm$0.148 &  0.797$\pm$0.138 &  0.836$\pm$0.127 &  0.867$\pm$0.116 \\
             &     & 10\% &  0.198$\pm$0.098 &  0.340$\pm$0.126 &  0.460$\pm$0.154 &  0.576$\pm$0.160 &  0.671$\pm$0.157 &  0.749$\pm$0.146 &  0.808$\pm$0.132 &  0.853$\pm$0.117 &  0.889$\pm$0.103 &  0.916$\pm$0.089 \\
             &     & 15\% &  0.212$\pm$0.103 &  0.362$\pm$0.129 &  0.489$\pm$0.159 &  0.614$\pm$0.159 &  0.711$\pm$0.151 &  0.789$\pm$0.134 &  0.845$\pm$0.116 &  0.887$\pm$0.099 &  0.920$\pm$0.084 &  0.944$\pm$0.069 \\
$deg^{rev}$ & 0.1 & 5\% &  0.076$\pm$0.052 &  0.136$\pm$0.065 &  0.191$\pm$0.084 &  0.240$\pm$0.104 &  0.285$\pm$0.120 &  0.324$\pm$0.131 &  0.361$\pm$0.140 &  0.396$\pm$0.148 &  0.430$\pm$0.154 &  0.465$\pm$0.160 \\
             &     & 10\% &  0.086$\pm$0.056 &  0.153$\pm$0.069 &  0.216$\pm$0.093 &  0.271$\pm$0.116 &  0.318$\pm$0.132 &  0.360$\pm$0.142 &  0.402$\pm$0.148 &  0.441$\pm$0.155 &  0.478$\pm$0.162 &  0.518$\pm$0.167 \\
             &     & 15\% &  0.094$\pm$0.059 &  0.166$\pm$0.071 &  0.235$\pm$0.099 &  0.293$\pm$0.125 &  0.344$\pm$0.141 &  0.389$\pm$0.147 &  0.434$\pm$0.152 &  0.476$\pm$0.159 &  0.516$\pm$0.167 &  0.562$\pm$0.170 \\
             & 0.3 & 5\% &  0.179$\pm$0.091 &  0.303$\pm$0.121 &  0.411$\pm$0.146 &  0.512$\pm$0.158 &  0.601$\pm$0.161 &  0.677$\pm$0.158 &  0.739$\pm$0.150 &  0.790$\pm$0.140 &  0.830$\pm$0.129 &  0.863$\pm$0.118 \\
             &     & 10\% &  0.202$\pm$0.099 &  0.345$\pm$0.126 &  0.465$\pm$0.152 &  0.580$\pm$0.159 &  0.675$\pm$0.155 &  0.752$\pm$0.144 &  0.811$\pm$0.129 &  0.856$\pm$0.114 &  0.892$\pm$0.100 &  0.919$\pm$0.086 \\
             &     & 15\% &  0.222$\pm$0.104 &  0.376$\pm$0.129 &  0.505$\pm$0.158 &  0.630$\pm$0.157 &  0.728$\pm$0.146 &  0.805$\pm$0.127 &  0.860$\pm$0.108 &  0.901$\pm$0.090 &  0.931$\pm$0.075 &  0.953$\pm$0.061 \\
$PR^{manip}$ & 0.1 & 5\% &  0.080$\pm$0.050 &  0.134$\pm$0.067 &  0.184$\pm$0.089 &  0.228$\pm$0.110 &  0.268$\pm$0.127 &  0.303$\pm$0.138 &  0.335$\pm$0.145 &  0.366$\pm$0.151 &  0.395$\pm$0.158 &  0.425$\pm$0.164 \\
                           &     & 10\% &  0.094$\pm$0.054 &  0.154$\pm$0.069 &  0.212$\pm$0.095 &  0.263$\pm$0.118 &  0.308$\pm$0.134 &  0.348$\pm$0.141 &  0.385$\pm$0.146 &  0.421$\pm$0.151 &  0.457$\pm$0.158 &  0.493$\pm$0.164 \\
                           &     & 15\% &  0.102$\pm$0.057 &  0.167$\pm$0.072 &  0.231$\pm$0.101 &  0.286$\pm$0.125 &  0.335$\pm$0.138 &  0.377$\pm$0.143 &  0.418$\pm$0.146 &  0.458$\pm$0.152 &  0.498$\pm$0.160 &  0.540$\pm$0.165 \\
                           & 0.3 & 5\% &  0.183$\pm$0.094 &  0.306$\pm$0.126 &  0.414$\pm$0.153 &  0.512$\pm$0.165 &  0.596$\pm$0.168 &  0.665$\pm$0.164 &  0.720$\pm$0.156 &  0.763$\pm$0.147 &  0.797$\pm$0.137 &  0.824$\pm$0.127 \\
                           &     & 10\% &  0.213$\pm$0.101 &  0.352$\pm$0.130 &  0.473$\pm$0.156 &  0.583$\pm$0.163 &  0.672$\pm$0.158 &  0.743$\pm$0.147 &  0.795$\pm$0.134 &  0.834$\pm$0.121 &  0.863$\pm$0.108 &  0.886$\pm$0.097 \\
                           &     & 15\% &  0.233$\pm$0.107 &  0.381$\pm$0.132 &  0.511$\pm$0.158 &  0.632$\pm$0.157 &  0.724$\pm$0.146 &  0.793$\pm$0.129 &  0.841$\pm$0.112 &  0.875$\pm$0.098 &  0.900$\pm$0.085 &  0.919$\pm$0.075 \\
\bottomrule
\end{tabular}
}
\caption{Results about the normalized $\Delta MoV$ of the algorithms under the \emph{limited-knowledge} hypothesis for 20 voters, $\NOISE=0$, 5 candidates. Each cell shows the average performance and the standard deviation. $PR1^+$, $N2^{rev}$, $deg^{rev}$, and $PR^{manip}$ respectively stand for \emph{SPpagerank1.0\_pos}, \emph{SPneig2\_rev}, \emph{SPoutdeg\_rev}, and \emph{SPpagerank1.0\_manip\_eq1}.}
\label{tab:results_wattsstrogatz_LK_STUDYING_STEP_2}
\end{table}
% From the plots, it is hard to distinguish the best algorithm as they are all equally good. However, w
We can check (see also numerical results in Table \ref{tab:results_wattsstrogatz_LK_STUDYING_STEP_2}) that the algorithm with the best performances is \emph{SPpagerank1.0\_pos}, because it considers the whole network, while almost every other ones focus on local properties of the graph.
% Algorithms averaging the scores $s_P$ and $s_G$ are excluded from the analysis since, from the previous considerations, it seems that $s_G$ is not important for manipulation purposes.
% Results show that \emph{SPpagerank1.0\_pos} sometimes gets slightly lower variance of the $\Delta MoV$. Generally, the four algorithms in Table \ref{tab:results_wattsstrogatz_LK_STUDYING_STEP_2} provide similar standard deviations. We can conclude that the best algorithm is \emph{SPpagerank1.0\_pos}.
%
% -EXP2: tuning dei parametri di PageRank (e.g Fig 4.4)
% \paragraph{Focusing on PageRank based heuristics.}
To stress this aspect, we next provide further comparisons between PageRank based heuristics.

Let us compare \emph{SPpagerank1.0\_pos} with other PageRank based heuristics. In particular, we will focus on \emph{SPpagerank1.0\_manip\_eq1}, that somehow uses the same weights as the ones defined in the greedy algorithm proposed in Section~\ref{sec:greedy}, i.e., considers only nodes that will vote for $c^*$ if manipulated and do not already vote for $c^*$. One may hope that \emph{SPpagerank1.0\_manip\_eq1} should inherit some good properties of the approximation algorithm and outperform \emph{SPpagerank1.0\_pos}. Our results (see Table~\ref{tab:results_wattsstrogatz_LK_STUDYING_STEP_2}) show that \emph{SPpagerank1.0\_manip\_eq1} actually achieves better performances in the first manipulation campaigns; after a few manipulations, \emph{SPpagerank1.0\_pos} achieves higher $\Delta MoV$s.
This means that if the manipulator has a low budget, then he should use \emph{SPpagerank1.0\_manip\_eq1} to achieve the best results (at least considering only the heuristics analyzed so far). Moreover, the additional cost required to estimate if a voter is manipulable does not increase the computational complexity concerning \emph{SPpagerank1.0\_pos}.

The fact that \emph{SPpagerank1.0\_manip\_eq1} performs better only for the first rounds suggests that it is ``too good'' at hiring the best influencers as soon as possible, but this leads to conditions in which the manipulation problem is harder to solve. The problem is that the algorithm is excessively eager to increase the $\Delta MoV$ in the first campaigns and does not consider the possibility of influencing the electorate in a higher number of campaigns. Algorithm \emph{SPpagerank1.0\_manip*\_pos} was specifically designed to overcome this issue. In fact, when computing the distance function (as shown in the definition of the algorithm), the weights of non-manipulable voters consider how much the voter will get closer to $c^*$ (on the political axis) if influenced; this property should give the algorithm the capability of foreseeing the manipulated electorate at the following manipulation campaign.
Of course, these are only intuitive conjectures. However, since the algorithms are based on heuristics, there is no way to formally prove the efficacy of these intuitions.
The comparison of the algorithms \emph{SPpagerank1.0\_pos}, \emph{SPpagerank1.0\_manip\_eq1}, and \emph{SPpagerank1.0\_manip*\_pos}  is shown in Figure \ref{fig:results_wattsstrogatz_LK_STUDYING_STEP_6}.
\begin{figure}
    \centering
    \includegraphics[width=0.9\textwidth]{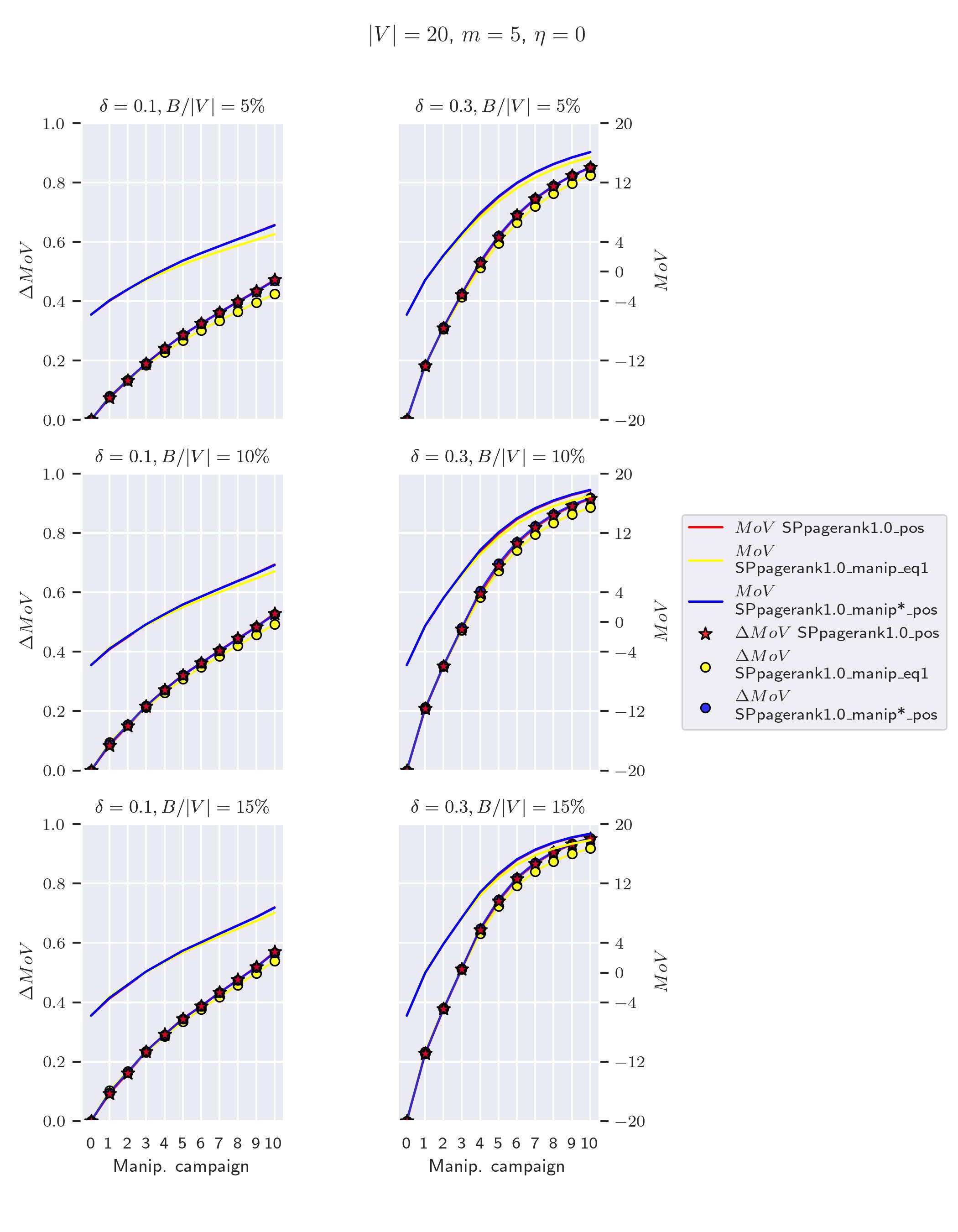}
    \caption{Performances of \emph{SPpagerank1.0\_pos} against \emph{SPpagerank1.0\_manip\_eq1} and \emph{SPpagerank1.0\_manip*\_pos} with 20 voters, $\NOISE=0$, and 5 candidates. Each plot shows both the $\Delta MoV$ and the $MoV$.}
\label{fig:results_wattsstrogatz_LK_STUDYING_STEP_6}
\end{figure}
Results show that \emph{SPpagerank1.0\_manip*\_pos} performs better than \emph{SPpagerank1.0\_manip\_eq1}. Anyway, it does not outperform \emph{SPpagerank1.0\_pos}. Since \emph{SPpagerank1.0\_manip*\_pos} is computationally more expensive, we conclude that the best algorithm using fast heuristics in the perfectly single-peaked scenario is \emph{SPpagerank1.0\_pos}.

% approximation algorithm
% -EXP3: confronto di fine-tuned PageRank con l'algoritmo di apx (Fig 4.5)
% \paragraph{Comparing with the Greedy Algorithm.}
We next compare the best heuristic method and the approximation algorithm.
% , which is the only one for which a formal analysis can be shown.
Since the approximation algorithm is computationally heavy (see below), the tests only use the following subset of parameters:
$\delta \in \{0.1, 0.3\}$ and
% considering the general trend for which increasing $\delta$ causes $\Delta MoV$ to increase;
$B/|V| \in \{5\%, 10\%\}$.
% , a value that is neither too low nor too high.
Moreover, in the 75\% of the considered instances the approximation algorithm is guaranteed to achieve a constant approximation (i.e., in these instances no candidate is more advantaged than the target candidate by messages in favour of the latter).
Performances are graphically shown in Figure~\ref{fig:results_wattsstrogatz_LK_STUDYING_STEP_9} (see also Table~\ref{tab:results_wattsstrogatz_LK_STUDYING_STEP_9}).
\begin{figure}
    \centering
    \includegraphics[width=0.75\textwidth]{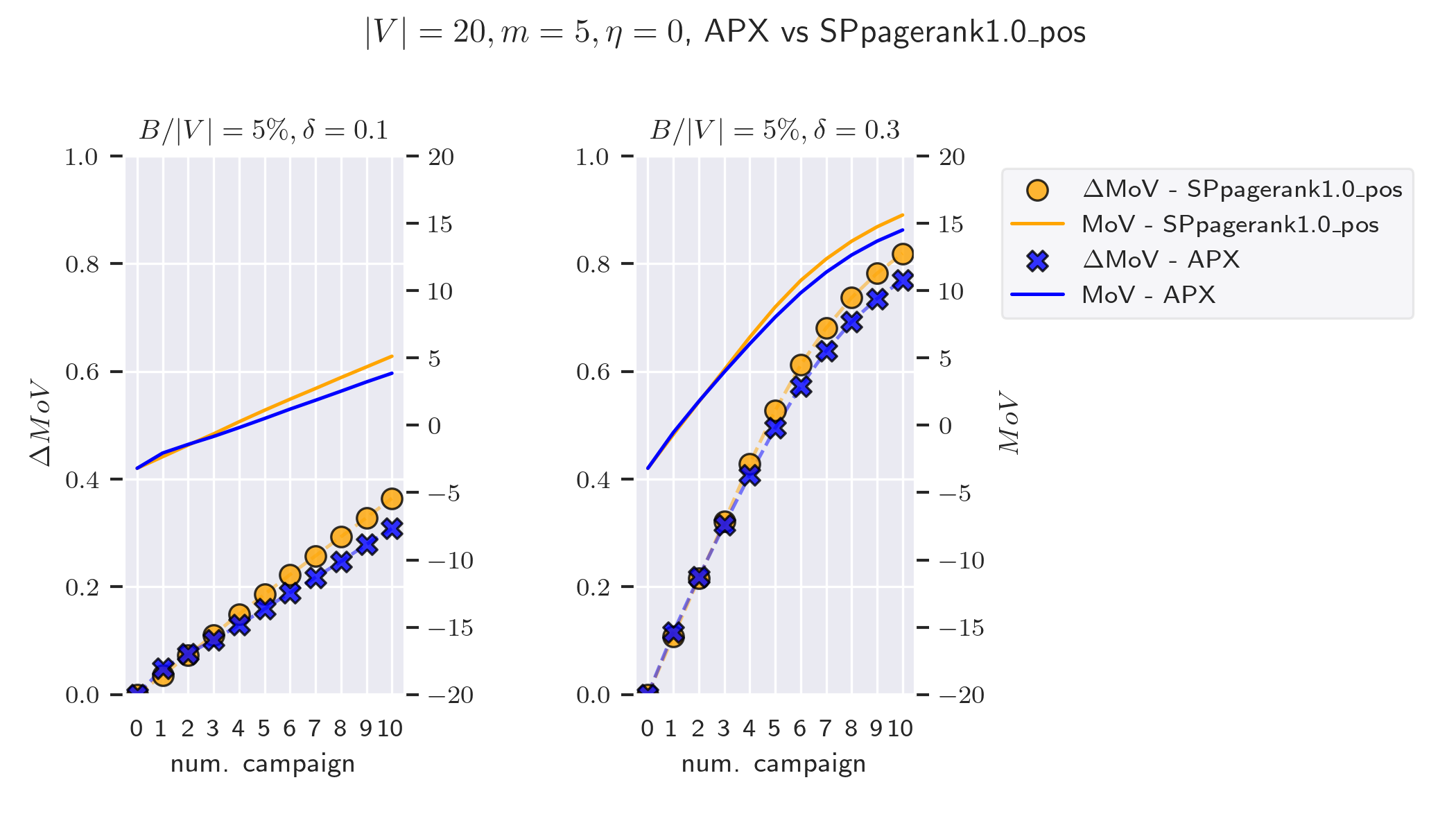}
    \caption{Performances of \emph{SPpagerank1.0\_pos} and of the approximation algorithm.}
%     The plots show both the $\Delta MoV$ and the $MoV$.}
\label{fig:results_wattsstrogatz_LK_STUDYING_STEP_9}
% \vspace{-0.6cm}
\end{figure}
\begin{table}[H]
\centering
\resizebox{\textwidth}{!}{
% this comes from pandas
\begin{tabular}{lllllllllllll}
\toprule
                  &     &     &       Round 1 &       Round 2 &       Round 3 &       Round 4 &       Round 5 &       Round 6 &       Round 7 &       Round 8 &       Round 9 &      Round 10 \\
Alg & $\delta$ & B &               &               &               &               &               &               &               &               &               &               \\
\midrule
APX & 0.1 & 5\% &  0.048$\pm$0.043 &  0.076$\pm$0.053 &  0.101$\pm$0.059 &  0.129$\pm$0.068 &  0.159$\pm$0.078 &  0.189$\pm$0.091 &  0.217$\pm$0.105 &  0.247$\pm$0.116 &  0.279$\pm$0.124 &  0.308$\pm$0.135 \\
                  &     & 10\% &  0.051$\pm$0.058 &  0.088$\pm$0.063 &  0.123$\pm$0.065 &  0.163$\pm$0.081 &  0.201$\pm$0.094 &  0.245$\pm$0.102 &  0.282$\pm$0.118 &  0.321$\pm$0.129 &  0.361$\pm$0.141 &  0.402$\pm$0.151 \\
                  & 0.3 & 5\% &  0.116$\pm$0.050 &  0.218$\pm$0.084 &  0.315$\pm$0.111 &  0.408$\pm$0.133 &  0.495$\pm$0.146 &  0.573$\pm$0.151 &  0.638$\pm$0.151 &  0.692$\pm$0.145 &  0.735$\pm$0.138 &  0.770$\pm$0.131 \\
                  &     & 10\% &  0.155$\pm$0.058 &  0.283$\pm$0.099 &  0.403$\pm$0.127 &  0.522$\pm$0.143 &  0.628$\pm$0.141 &  0.714$\pm$0.131 &  0.777$\pm$0.119 &  0.821$\pm$0.107 &  0.852$\pm$0.097 &  0.875$\pm$0.090 \\
$PR1^+$ & 0.1 & 5\% &  0.036$\pm$0.047 &  0.073$\pm$0.057 &  0.111$\pm$0.060 &  0.149$\pm$0.071 &  0.186$\pm$0.084 &  0.223$\pm$0.096 &  0.257$\pm$0.107 &  0.293$\pm$0.116 &  0.328$\pm$0.124 &  0.364$\pm$0.133 \\
                  &     & 10\% &  0.040$\pm$0.051 &  0.084$\pm$0.058 &  0.127$\pm$0.063 &  0.172$\pm$0.075 &  0.214$\pm$0.092 &  0.255$\pm$0.104 &  0.295$\pm$0.116 &  0.337$\pm$0.125 &  0.378$\pm$0.133 &  0.420$\pm$0.142 \\
                  & 0.3 & 5\% &  0.108$\pm$0.055 &  0.215$\pm$0.092 &  0.321$\pm$0.117 &  0.429$\pm$0.136 &  0.528$\pm$0.144 &  0.612$\pm$0.145 &  0.681$\pm$0.140 &  0.737$\pm$0.133 &  0.782$\pm$0.124 &  0.819$\pm$0.114 \\
                  &     & 10\% &  0.124$\pm$0.061 &  0.250$\pm$0.103 &  0.374$\pm$0.126 &  0.497$\pm$0.139 &  0.607$\pm$0.139 &  0.698$\pm$0.133 &  0.768$\pm$0.123 &  0.822$\pm$0.112 &  0.864$\pm$0.100 &  0.895$\pm$0.089 \\
\bottomrule
\end{tabular}
% \begin{tabular}{lllllllllllll}
% \toprule
%             &     &     &       Round 1 &       Round 2 &       Round 3 &       Round 4 &       Round 5 &       Round 6 &       Round 7 &       Round 8 &       Round 9 &      Round 10 \\
% Alg & $\delta$ & B &               &               &               &               &               &               &               &               &               &               \\
% \midrule
% $PR1^+$ & 0.3 & 10\% &  0.166$\pm$0.044 &  0.307$\pm$0.058 &  0.450$\pm$0.079 &  0.564$\pm$0.088 &  0.667$\pm$0.091 &  0.747$\pm$0.090 &  0.809$\pm$0.086 &  0.857$\pm$0.080 &  0.894$\pm$0.074 &  0.922$\pm$0.067 \\
% $APX$ & 0.3 & 10\% &  0.183$\pm$0.039 &  0.320$\pm$0.060 &  0.457$\pm$0.080 &  0.569$\pm$0.095 &  0.662$\pm$0.097 &  0.733$\pm$0.094 &  0.786$\pm$0.089 &  0.823$\pm$0.085 &  0.851$\pm$0.081 &  0.871$\pm$0.079 \\
% \bottomrule
% \end{tabular}
}
\caption{Results about the normalized $\Delta MoV$ of the algorithms under the \emph{limited-knowledge} hypothesis for 20 voters, $\NOISE=0$, 5 candidates. Each cell shows the average performance and the standard deviation. $PR1^+$ and $APX$ respectively stand for \emph{SPpagerank1.0\_pos} and \emph{Approximation algorithm}.}
\label{tab:results_wattsstrogatz_LK_STUDYING_STEP_9}
\end{table}
Results show that the approximation algorithm actually performs better than \emph{SPpagerank1.0\_pos} only in the initial campaigns.
% This result is similar to the one of \emph{SPpagerank1.0\_manip\_eq1} and should directly follow from the fact that they use the same weights for the nodes.
For a higher number of rounds, \emph{SPpagerank1.0\_pos} performs better, and on average, after ten campaigns, it gets approximately more votes for the target candidate
% than the approximation algorithm
(see Table \ref{tab:apx_vs_SPpagerank1pos} for a more detailed comparison).
% of the two algorithms).
\begin{table}[H]
\centering
\resizebox{\textwidth}{!}{
\begin{tabular}{lllllllllllll}
\toprule
    &     &                   & Round 1 &  R. 2 &  R. 3 &  R. 4 &  R. 5 &  R. 6 &  R. 7 &  R. 8 &  R. 9 & R. 10 \\
B & $\delta$ & {} &         &       &       &       &       &       &       &       &       &       \\
\midrule
5\% & 0.1 & \# $(APX > PR)\ \%$ &    62\ \% &  52\ \% &  39\ \% &  25\ \% &  21\ \% &  19\ \% &  19\ \% &  16\ \% &  16\ \% &  18\ \% \\
    &     & max $(APX-PR)\ \%$ &    10\ \% &  20\ \% &  11\ \% &  13\ \% &  11\ \% &  13\ \% &  13\ \% &  12\ \% &  16\ \% &  18\ \% \\
    &     & mean $(APX-PR)\ \%$ &     1\ \% &   0\ \% &   0\ \% &  -2\ \% &  -2\ \% &  -3\ \% &  -3\ \% &  -4\ \% &  -4\ \% &  -5\ \% \\
    & 0.3 & \# $(APX > PR)\ \%$ &    54\ \% &  53\ \% &  43\ \% &  29\ \% &  19\ \% &  12\ \% &   9\ \% &   8\ \% &   6\ \% &   6\ \% \\
    &     & max $(APX-PR)\ \%$ &    15\ \% &  16\ \% &  15\ \% &   9\ \% &   7\ \% &   6\ \% &   7\ \% &   7\ \% &   8\ \% &   9\ \% \\
    &     & mean $(APX-PR)\ \%$ &     0\ \% &   0\ \% &   0\ \% &  -2\ \% &  -3\ \% &  -3\ \% &  -4\ \% &  -4\ \% &  -4\ \% &  -4\ \% \\
10\% & 0.1 & \# $(APX > PR)\ \%$ &    57\ \% &  58\ \% &  51\ \% &  45\ \% &  43\ \% &  44\ \% &  45\ \% &  42\ \% &  45\ \% &  46\ \% \\
    &     & max $(APX-PR)\ \%$ &    15\ \% &  13\ \% &  11\ \% &  14\ \% &   9\ \% &  19\ \% &  18\ \% &  21\ \% &  24\ \% &  26\ \% \\
    &     & mean $(APX-PR)\ \%$ &     1\ \% &   0\ \% &   0\ \% &   0\ \% &  -1\ \% &  -1\ \% &  -1\ \% &  -1\ \% &  -1\ \% &  -1\ \% \\
    & 0.3 & \# $(APX > PR)\ \%$ &    83\ \% &  84\ \% &  80\ \% &  77\ \% &  74\ \% &  71\ \% &  64\ \% &  52\ \% &  36\ \% &  25\ \% \\
    &     & max $(APX-PR)\ \%$ &    23\ \% &  19\ \% &  19\ \% &  20\ \% &  15\ \% &  18\ \% &  16\ \% &  12\ \% &  11\ \% &  10\ \% \\
    &     & mean $(APX-PR)\ \%$ &     3\ \% &   3\ \% &   2\ \% &   2\ \% &   2\ \% &   1\ \% &   0\ \% &   0\ \% &  -1\ \% &  -2\ \% \\
\bottomrule
\end{tabular}
% \begin{tabular}{lllllllllll}
% \toprule
% {} & Round 1 &  R. 2 &  R. 3 &  R. 4 &  R. 5 &  R. 6 &  R. 7 &  R. 8 &  R. 9 & R. 10 \\
% \midrule
% \ \# $(APX > PR)\%$ &    65\ \% &  57\ \% &  55\ \% &  53\ \% &  50\ \% &  48\ \% &  42\ \% &  31\ \% &  25\ \% &  23\ \% \\
% $max(APX-PR)\ \%$      &    12\ \% &  18\ \% &  17\ \% &  25\ \% &  22\ \% &  17\ \% &  16\ \% &  14\ \% &  12\ \% &  10\ \% \\
% $avg(APX-PR)\ \%$     &     1\ \% &   1\ \% &   0\ \% &   0\ \% &   0\ \% &  -1\ \% &  -2\ \% &  -3\ \% &  -4\ \% &  -5\ \% \\
% \bottomrule
% \end{tabular}
}
\caption{Detailed comparison of the approximation algorithm (named $APX$) and \emph{SPpagerank1.0\_pos} (named $PR$). The experiments are the same as in Table \ref{tab:results_wattsstrogatz_LK_STUDYING_STEP_9}. The first row of the table shows the number of times the approximation algorithm achieves better (normalized) $\Delta MoV$ than \emph{SPpagerank1.0\_pos}. The second and third rows show the maximum and average difference of the normalized $\Delta MoV$ between the approximation algorithm and the heuristic considering the 775 elections they were tested on.}
\label{tab:apx_vs_SPpagerank1pos}
\end{table}

% We can see that the approximation algorithm is better than the heuristic in up to 65\% of the tests, but the average difference is less than 1\% for the first campaigns, and for higher campaigns, it is overwhelmed by \emph{SPpagerank1.0\_pos}.
% In conclusion, at least in these settings, the approximation algorithm is not worthwhile, and timing experiments in the following sections will confirm this thesis.

% -EXP4: Robustness of PageRank to (different types of) noise (Fig 4.6 - 4.8)
% \paragraph{Nearly Single-Peaked Scenarios.}
Results in Figure \ref{fig:results_wattsstrogatz_LK_STUDYING_STEP_3} (and in Table \ref{tab:results_wattsstrogatz_LK_STUDYING_STEP_3})
% shows the numerical average of the normalized $\Delta MoV$ and the standard deviations).
show the performances of algorithm \emph{SPpagerank1.0\_pos} in nearly single-peaked scenarios.
% The downside of these simulations is that the algorithm cannot be tested in the same initial conditions for different noises $\NOISE$. In fact,
Note that the blurred views of the voters make the initial MoVs of the simulations different from the initial MoVs in the perfectly single-peaked cases.
% (one way to reduce this effect is shown below).
For this reason, comparisons based on $MoV$ cannot be made.
% (although the margin of victory is shown for completeness).
% Moreover, also comparisons based on the normalized $\Delta MoV$ can be misleading
%
% commentare i risultati di medie e varianze
% The plots display weird results: in some cases,
While it may appear that blurred views increase the performances,
% . This is totally counter-intuitive since, in non-single-peaked electorates, the manipulator estimates wrong weights for the voters. This is the result of the above considerations:
this only depends on
the initial conditions of the simulated scenarios
being different.
% are different due to the blur; hence, the electorate evolves in different ways depending on the noise.
However, this only happens when the noise has low variance: with $\NOISE=\mathcal{N}(0;1)$, performances are definitely worse than the single-peaked case (even if the initial MoV is higher). Nevertheless, even with
% non è andato tanto male, anche considerando che il noise è grande
% Considering that the political spectrum has length 2, both
$\NOISE=\mathcal{N}(0;1)$ and $\NOISE=\mathcal{N}(0;0.08)$
that
are very strong noises,
% . Nevertheless,
performances are not that much worse than the ones in single-peaked scenarios.
% How is it possible? One possible explanation is the following. Since $x^v_c = x_c + \mathcal{N}(0;\cdot)$, then $x^v_c$ is Gaussian and its distribution $f_{x^v_c}$ has the peak in $x_c$. When the manipulator computes $d_P(u,c^*)=|x_u - x_{c^*}|$, he is not just using the real position of the target candidate; by assuming that $x^v_c=x_c$, he is also maximizing the prior probability $f_{x^v_c}$.
% To prove this hypothesis, the
For this reason, our heuristics
% algorithm can be
has been
tested also with specific noises not having the peak of the probability distribution function in 0, namely $\NOISE_2=\frac{1}{2}\mathcal{N}(-0.7;1)+\frac{1}{2}\mathcal{N}(+0.7;1)$.
The corresponding plots in Figure \ref{fig:results_wattsstrogatz_LK_STUDYING_STEP_3}
show
% confirm the above considerations, showing
that performances slightly drop when the peak of the distribution of the noise is not 0.
% Moreover, since the distributions of the $k$-swaps are similar, simulations start with almost the same average initial MoV; this makes the results more comparable than the ones in Figure \ref{fig:results_wattsstrogatz_LK_STUDYING_STEP_3}.
% \begin{figure}
%     \centering
%     \includegraphics[width=1.0\textwidth]{results_wattsstrogatz_LK_STUDYING_STEP_5.jpg}
%     \caption{Performances of \emph{SPpagerank1.0\_pos} for $\NOISE=\mathcal{N}(0;1)$ and $\NOISE=\frac{1}{2}\mathcal{N}(-0.7;1)+\frac{1}{2}\mathcal{N}(+0.7;1)$.}
% %     The plots show both the normalized $\Delta MoV$ and the $MoV$.}
% \label{fig:results_wattsstrogatz_LK_STUDYING_STEP_5}
% \end{figure}
%
% In conclusion, the noise confuses the manipulator: the higher the variance, the lower his power. Moreover, the shape of the distribution can alter the results of the manipulation, too. Hence, we conclude that, in this particular manipulation model, perfectly single-peaked electorates are more easily manipulable than nearly-single-peaked ones.
\begin{figure}
% \vspace{-0.5cm}
    \centering
    \includegraphics[width=0.75\textwidth]{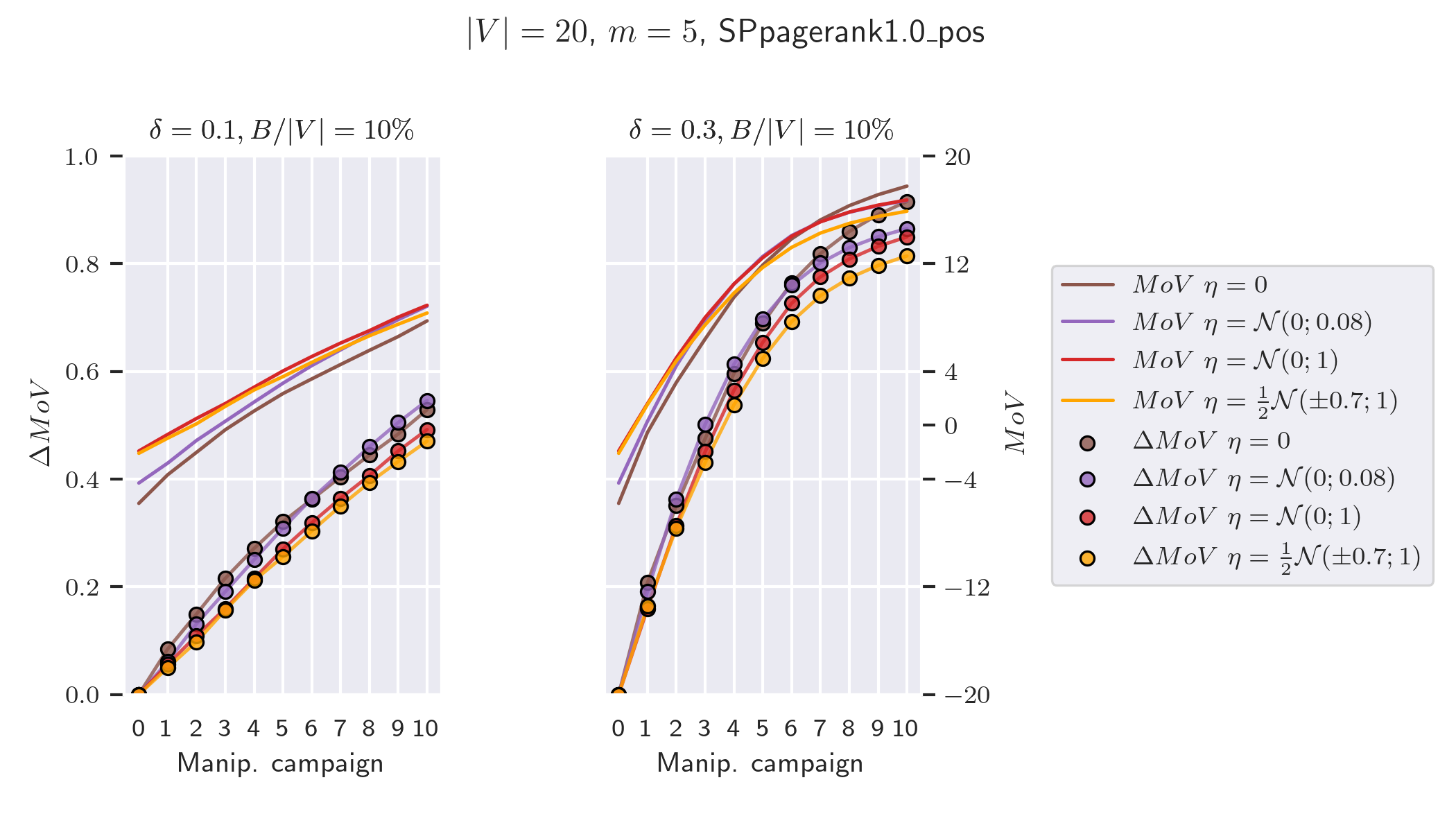}
    \caption{Performances of \emph{SPpagerank1.0\_pos} for $\NOISE=\mathcal{N}(0;0.08)$ and $\NOISE=\mathcal{N}(0;1)$ and $\NOISE=\frac{1}{2}\mathcal{N}(-0.7;1)+\frac{1}{2}\mathcal{N}(+0.7;1)$.}
    %     The plots show both the normalized $\Delta MoV$ and the $MoV$, but comparisons based on $MoV$ do not make sense. Comparisons based on $\Delta MoV$ can be misleading, too.}
\label{fig:results_wattsstrogatz_LK_STUDYING_STEP_3}
% \vspace{-0.6cm}
\end{figure}
\begin{table}[H]
\centering
\resizebox{\textwidth}{!}{
\begin{tabular}{lllllllllllll}
\toprule
                    &     &     &       Round 1 &       Round 2 &       Round 3 &       Round 4 &       Round 5 &       Round 6 &       Round 7 &       Round 8 &       Round 9 &      Round 10 \\
$\eta$ & $\delta$ & B &               &               &               &               &               &               &               &               &               &               \\
\midrule
0 & 0.1 & 5\% &  0.074$\pm$0.053 &  0.133$\pm$0.068 &  0.190$\pm$0.090 &  0.240$\pm$0.112 &  0.286$\pm$0.127 &  0.325$\pm$0.135 &  0.362$\pm$0.142 &  0.399$\pm$0.149 &  0.434$\pm$0.157 &  0.472$\pm$0.162 \\
                    &     & 10\% &  0.084$\pm$0.057 &  0.149$\pm$0.071 &  0.216$\pm$0.097 &  0.271$\pm$0.122 &  0.321$\pm$0.135 &  0.362$\pm$0.141 &  0.404$\pm$0.145 &  0.445$\pm$0.152 &  0.484$\pm$0.161 &  0.529$\pm$0.164 \\
                    &     & 15\% &  0.091$\pm$0.059 &  0.161$\pm$0.073 &  0.233$\pm$0.104 &  0.291$\pm$0.130 &  0.344$\pm$0.143 &  0.387$\pm$0.144 &  0.433$\pm$0.148 &  0.476$\pm$0.156 &  0.519$\pm$0.167 &  0.570$\pm$0.168 \\
                    & 0.3 & 5\% &  0.182$\pm$0.094 &  0.310$\pm$0.124 &  0.421$\pm$0.153 &  0.528$\pm$0.160 &  0.616$\pm$0.160 &  0.689$\pm$0.154 &  0.745$\pm$0.146 &  0.788$\pm$0.136 &  0.823$\pm$0.126 &  0.851$\pm$0.116 \\
                    &     & 10\% &  0.208$\pm$0.102 &  0.351$\pm$0.128 &  0.475$\pm$0.156 &  0.596$\pm$0.153 &  0.690$\pm$0.146 &  0.765$\pm$0.132 &  0.819$\pm$0.119 &  0.860$\pm$0.107 &  0.891$\pm$0.095 &  0.916$\pm$0.085 \\
                    &     & 15\% &  0.226$\pm$0.108 &  0.377$\pm$0.131 &  0.511$\pm$0.162 &  0.644$\pm$0.150 &  0.740$\pm$0.138 &  0.816$\pm$0.116 &  0.867$\pm$0.099 &  0.905$\pm$0.085 &  0.932$\pm$0.072 &  0.952$\pm$0.061 \\
$\mathcal{N}(0;1)$ & 0.1 & 5\% &  0.051$\pm$0.058 &  0.099$\pm$0.081 &  0.146$\pm$0.096 &  0.195$\pm$0.108 &  0.246$\pm$0.118 &  0.292$\pm$0.127 &  0.335$\pm$0.137 &  0.375$\pm$0.143 &  0.417$\pm$0.146 &  0.454$\pm$0.149 \\
                    &     & 10\% &  0.056$\pm$0.062 &  0.109$\pm$0.085 &  0.160$\pm$0.101 &  0.215$\pm$0.112 &  0.270$\pm$0.120 &  0.319$\pm$0.129 &  0.365$\pm$0.136 &  0.407$\pm$0.140 &  0.452$\pm$0.140 &  0.492$\pm$0.142 \\
                    &     & 15\% &  0.059$\pm$0.064 &  0.114$\pm$0.088 &  0.168$\pm$0.104 &  0.227$\pm$0.114 &  0.285$\pm$0.121 &  0.335$\pm$0.129 &  0.382$\pm$0.136 &  0.426$\pm$0.139 &  0.475$\pm$0.137 &  0.515$\pm$0.139 \\
                    & 0.3 & 5\% &  0.145$\pm$0.090 &  0.286$\pm$0.121 &  0.414$\pm$0.138 &  0.521$\pm$0.145 &  0.608$\pm$0.145 &  0.679$\pm$0.139 &  0.732$\pm$0.132 &  0.771$\pm$0.125 &  0.799$\pm$0.118 &  0.821$\pm$0.112 \\
                    &     & 10\% &  0.160$\pm$0.096 &  0.314$\pm$0.124 &  0.451$\pm$0.133 &  0.565$\pm$0.134 &  0.654$\pm$0.131 &  0.727$\pm$0.120 &  0.776$\pm$0.112 &  0.809$\pm$0.106 &  0.832$\pm$0.101 &  0.850$\pm$0.096 \\
                    &     & 15\% &  0.167$\pm$0.099 &  0.331$\pm$0.125 &  0.473$\pm$0.131 &  0.591$\pm$0.129 &  0.682$\pm$0.124 &  0.754$\pm$0.110 &  0.801$\pm$0.103 &  0.830$\pm$0.098 &  0.851$\pm$0.094 &  0.866$\pm$0.090 \\
$\mathcal{N}(0;0.08)$ & 0.1 & 5\% &  0.055$\pm$0.056 &  0.118$\pm$0.074 &  0.173$\pm$0.089 &  0.227$\pm$0.104 &  0.282$\pm$0.114 &  0.334$\pm$0.124 &  0.380$\pm$0.132 &  0.425$\pm$0.140 &  0.468$\pm$0.146 &  0.507$\pm$0.149 \\
                    &     & 10\% &  0.061$\pm$0.060 &  0.131$\pm$0.079 &  0.191$\pm$0.095 &  0.250$\pm$0.109 &  0.309$\pm$0.115 &  0.364$\pm$0.125 &  0.412$\pm$0.132 &  0.460$\pm$0.138 &  0.505$\pm$0.143 &  0.546$\pm$0.144 \\
                    &     & 15\% &  0.064$\pm$0.062 &  0.139$\pm$0.082 &  0.201$\pm$0.098 &  0.264$\pm$0.112 &  0.325$\pm$0.116 &  0.382$\pm$0.126 &  0.432$\pm$0.132 &  0.482$\pm$0.137 &  0.528$\pm$0.142 &  0.570$\pm$0.142 \\
                    & 0.3 & 5\% &  0.172$\pm$0.085 &  0.332$\pm$0.118 &  0.463$\pm$0.138 &  0.571$\pm$0.145 &  0.654$\pm$0.140 &  0.719$\pm$0.131 &  0.765$\pm$0.122 &  0.798$\pm$0.114 &  0.823$\pm$0.106 &  0.842$\pm$0.100 \\
                    &     & 10\% &  0.192$\pm$0.092 &  0.363$\pm$0.120 &  0.502$\pm$0.135 &  0.613$\pm$0.138 &  0.697$\pm$0.128 &  0.760$\pm$0.115 &  0.802$\pm$0.105 &  0.830$\pm$0.098 &  0.850$\pm$0.093 &  0.865$\pm$0.088 \\
                    &     & 15\% &  0.202$\pm$0.095 &  0.381$\pm$0.120 &  0.525$\pm$0.135 &  0.639$\pm$0.135 &  0.724$\pm$0.121 &  0.786$\pm$0.105 &  0.824$\pm$0.095 &  0.849$\pm$0.089 &  0.867$\pm$0.084 &  0.880$\pm$0.081 \\
\bottomrule
\end{tabular}
}
\caption{Results about the normalized $\Delta MoV$ of algorithm \emph{SPpagerank1.0\_pos} under the \emph{limited-knowledge} hypothesis for 20 voters, 5 candidates, for single-peaked and nearly-single-peaked electorates. Each cell shows the average performance and the standard deviation.}
\label{tab:results_wattsstrogatz_LK_STUDYING_STEP_3}
\end{table}

% \paragraph{Results in other Scenarios.}
% -EXP5: Robustness to input difference (Fig 4.9 - 4.10 - 4.12 - 4.14)
Until now, the shown experiments only involved electorates made up of 20 voters. We now analyze how performances change when testing electorates of 20, 50, and 100 voters.
% There is one point to clarify. The best fast heuristic found for 20-voters electorates is \emph{SPpagerank1.0\_pos}; to compute the score of node $v$, it only considers direct friends of $v$. In the Watts-Strogatz models we are considering, with 20 voters, each node has three neighbours on average, that is 15\% of the electorate. Increasing the number of voters to 50, the average number of friends of each node is still 3 (remember that, in these experiments, the spatial density of the nodes is independent of the number of nodes), which is only 6\% of the nodes. With 100 voters, the neighbourhood only considers 3\% of the nodes. Hence, \emph{SPpagerank1.0\_pos} could end up using too limited information for each node, estimating weights that do not actually reflect the importance of the voters.
% For this reason, this experimental phase compares performances of \emph{SPpagerank1.0\_pos} and \emph{SPpagerank1.0\_hop2\_pos} to see whether the algorithm benefits from considering larger neighbourhoods. Since we are dealing with only two algorithms, standard deviations of the normalized $\Delta MoV$ can be clearly shown directly on the plots.
% To reduce the number of simulations, t
The algorithms were only tested with a single, medium budget: 10\% of the electorate. Moreover, the target candidate was fixed to the right-most one on the political spectrum.
% ; the reasons are explained below.
Figure \ref{fig:results_wattsstrogatz_LK_STUDYING_STEP_8} illustrates the results.
% mettere figura
\begin{figure}[ht]
% \vspace{-0.3cm}
    \centering
    \includegraphics[width=0.75\textwidth]{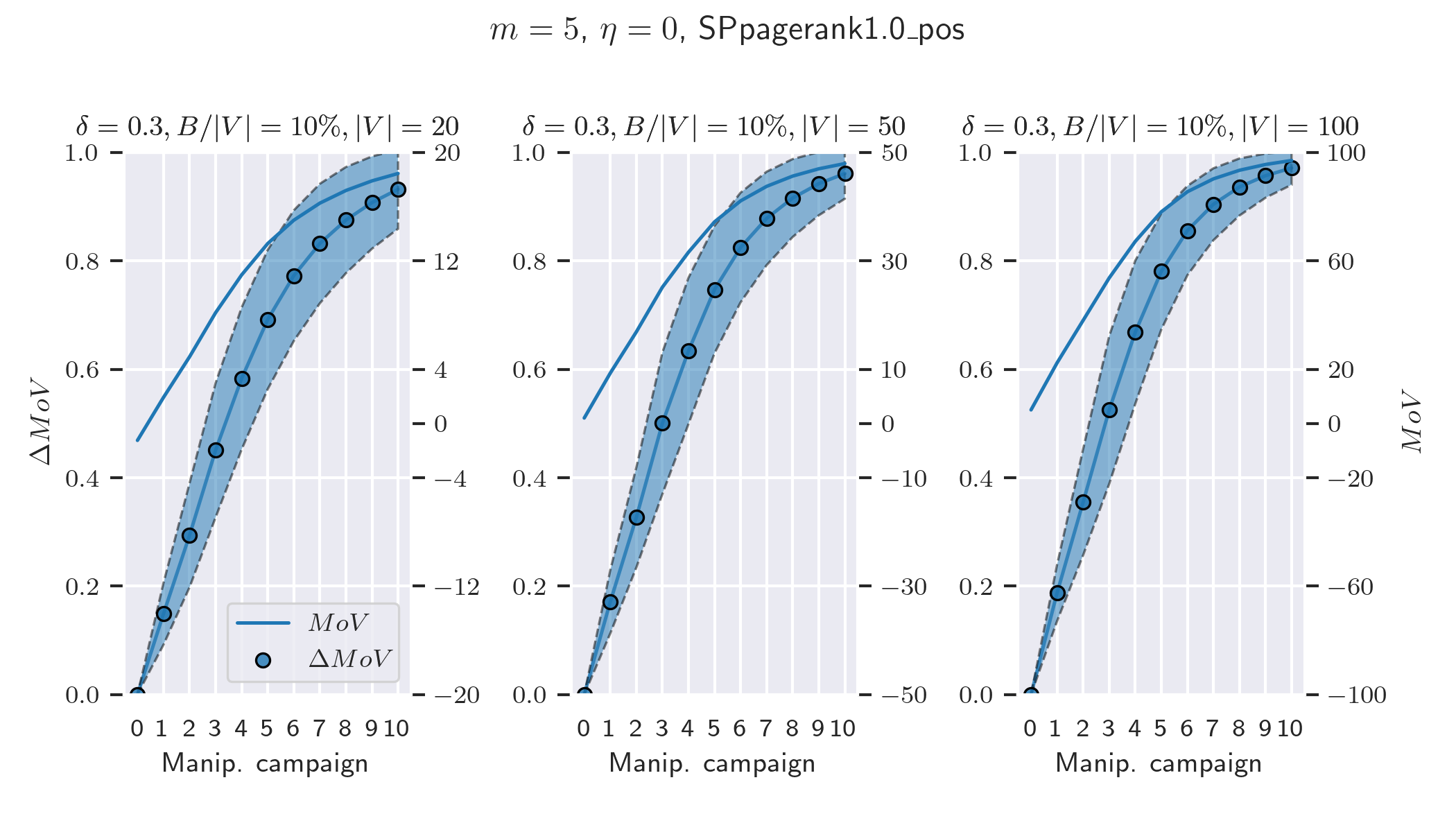}
    \caption{Performances of \emph{SPpagerank1.0\_pos} in perfectly single-peaked scenarios.}
%     The plots show both the normalized $\Delta MoV$ and the $MoV$. Shaded regions represent the standard deviation from the average normalized $\Delta MoV$.}
\label{fig:results_wattsstrogatz_LK_STUDYING_STEP_8}
% \vspace{-0.6cm}
\end{figure}
% dire che il deltamov normalizzato sale all'aumentare del numero di votanti.
% The first thing worth noting about the normalized $\Delta MoV$ is that
Note that performances increase when the number of voters increases.
% If this trend remained stable for even larger graphs, that would mean that real-world elections may be easily manipulable.
% hop2 non va meglio probabilmente xk pagerank già include la struttura della rete
% Moreover, we can conclude that increasing the neighbourhood does not increase the performance. How is it possible? Considering a larger neighbourhood should account for more information in the structure of the network. Nevertheless, by definition, PageRank values take into account the links among the nodes and the network in general. Hence, maybe the possibility of a node of influencing another node at distance 2 - as in \emph{SPpagerank1.0\_hop2\_pos} - is somehow considered when the simpler \emph{SPpagerank1.0\_pos} applies the update rule to share PageRank among neighbours, since sharing depends on the structure of the graph.
% confronto tra varianza
% In addition, when comparing the variances of the performances of the two algorithms, it is clear that they generally behave similarly, proving that increasing the size of the neighbourhood is not effective.

We next show that similar results hold also in the case of noisy single-peaked preferences.
Moreover, related experiments address an important issue arising when introducing noise. As visible in Figure \ref{fig:results_wattsstrogatz_LK_STUDYING_STEP_3}, noisy scenarios are characterized by higher initial MoVs. Trying to understand the reasons for this effect, results show that it partially depends on randomness; however, a simple way to reduce the effect is to change the target candidate. Remember that the previously shown experiments randomly select the target candidate. Instead, the following experiments use the right-most candidate as the target. In this way, the initial MoV does not change too much, and noisy and clear scenarios are more comparable, too. This partially solves the problem highlighted in the results of Figure \ref{fig:results_wattsstrogatz_LK_STUDYING_STEP_3}. Figure \ref{fig:results_wattsstrogatz_LK_STUDYING_STEP_8_noise} displays the tests in this new scenario; it is clear that the initial MoV in noisy electorates is almost identical to the one in single-peaked electorates, with a deviation of only one vote on average in the worst case for 20-voters electorates. The figure also shows the performances on larger graphs. In general, by changing the target candidate as described, performances with and without noise are much more divergent, proving that the manipulator struggles to increase the margin of victory in very noisy settings. This confirms that the election is easier to manipulate in perfectly single-peaked electorates (Table \ref{tab:results_wattsstrogatz_LK_STUDYING_STEP_8_noise} shows numerical average performances and standard deviations for such experiments).
\begin{figure}[H]
    \centering
    \includegraphics[width=1.0\textwidth]{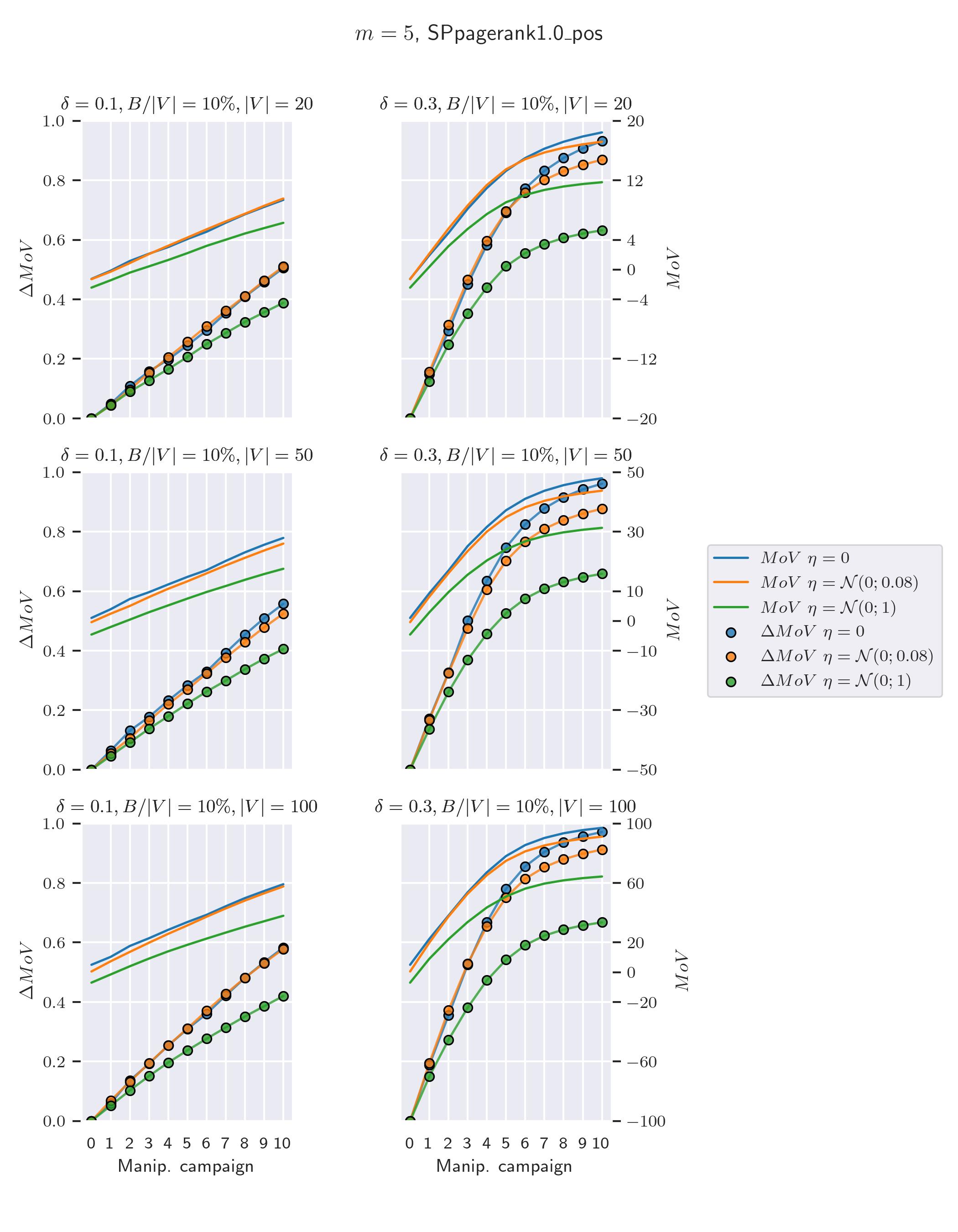}
    \caption{Performances of the algorithms \emph{SPpagerank1.0\_pos} in nearly-single-peaked scenarios. The plots show both the normalized $\Delta MoV$ and the $MoV$. The target candidate of the simulated elections is the right-most one on the political spectrum.}
\label{fig:results_wattsstrogatz_LK_STUDYING_STEP_8_noise}
\end{figure}
\begin{table}[H]
\centering
\resizebox{\textwidth}{!}{
\begin{tabular}{llllllllllllll}
\toprule
    &                    &     &     &       Round 1 &       Round 2 &       Round 3 &       Round 4 &       Round 5 &       Round 6 &       Round 7 &       Round 8 &       Round 9 &      Round 10 \\
$|V|$ & $\eta$ & $\delta$ & B &               &               &               &               &               &               &               &               &               &               \\
\midrule
20  & 0 & 0.1 & 10\% &  0.049$\pm$0.055 &  0.109$\pm$0.062 &  0.158$\pm$0.059 &  0.198$\pm$0.067 &  0.245$\pm$0.094 &  0.296$\pm$0.102 &  0.353$\pm$0.110 &  0.409$\pm$0.118 &  0.458$\pm$0.127 &  0.505$\pm$0.139 \\
    &                    & 0.3 & 10\% &  0.149$\pm$0.057 &  0.294$\pm$0.096 &  0.451$\pm$0.124 &  0.583$\pm$0.131 &  0.692$\pm$0.128 &  0.773$\pm$0.120 &  0.832$\pm$0.110 &  0.875$\pm$0.098 &  0.907$\pm$0.085 &  0.932$\pm$0.073 \\
    & $\mathcal{N}(0;0.08)$ & 0.1 & 10\% &  0.044$\pm$0.069 &  0.097$\pm$0.083 &  0.153$\pm$0.103 &  0.205$\pm$0.104 &  0.257$\pm$0.105 &  0.310$\pm$0.110 &  0.361$\pm$0.114 &  0.411$\pm$0.122 &  0.463$\pm$0.135 &  0.510$\pm$0.138 \\
    &                    & 0.3 & 10\% &  0.156$\pm$0.098 &  0.315$\pm$0.104 &  0.466$\pm$0.127 &  0.597$\pm$0.129 &  0.696$\pm$0.122 &  0.759$\pm$0.114 &  0.801$\pm$0.106 &  0.831$\pm$0.099 &  0.852$\pm$0.093 &  0.869$\pm$0.088 \\
    & $\mathcal{N}(0;1)$ & 0.1 & 10\% &  0.044$\pm$0.056 &  0.090$\pm$0.078 &  0.127$\pm$0.086 &  0.165$\pm$0.090 &  0.206$\pm$0.094 &  0.250$\pm$0.098 &  0.287$\pm$0.103 &  0.324$\pm$0.108 &  0.356$\pm$0.111 &  0.388$\pm$0.116 \\
    &                    & 0.3 & 10\% &  0.125$\pm$0.077 &  0.248$\pm$0.093 &  0.352$\pm$0.106 &  0.441$\pm$0.117 &  0.511$\pm$0.125 &  0.555$\pm$0.127 &  0.586$\pm$0.129 &  0.606$\pm$0.130 &  0.621$\pm$0.130 &  0.632$\pm$0.130 \\
50  & 0 & 0.1 & 10\% &  0.064$\pm$0.042 &  0.131$\pm$0.049 &  0.177$\pm$0.058 &  0.232$\pm$0.071 &  0.283$\pm$0.086 &  0.330$\pm$0.096 &  0.393$\pm$0.103 &  0.454$\pm$0.118 &  0.509$\pm$0.133 &  0.557$\pm$0.145 \\
    &                    & 0.3 & 10\% &  0.171$\pm$0.058 &  0.327$\pm$0.092 &  0.501$\pm$0.131 &  0.634$\pm$0.135 &  0.747$\pm$0.117 &  0.825$\pm$0.101 &  0.878$\pm$0.086 &  0.916$\pm$0.072 &  0.942$\pm$0.058 &  0.961$\pm$0.046 \\
    & $\mathcal{N}(0;0.08)$ & 0.1 & 10\% &  0.055$\pm$0.039 &  0.106$\pm$0.054 &  0.165$\pm$0.063 &  0.220$\pm$0.070 &  0.270$\pm$0.080 &  0.323$\pm$0.087 &  0.377$\pm$0.091 &  0.428$\pm$0.099 &  0.478$\pm$0.107 &  0.525$\pm$0.114 \\
    &                    & 0.3 & 10\% &  0.167$\pm$0.061 &  0.325$\pm$0.085 &  0.474$\pm$0.106 &  0.605$\pm$0.110 &  0.702$\pm$0.103 &  0.767$\pm$0.093 &  0.809$\pm$0.084 &  0.839$\pm$0.077 &  0.861$\pm$0.071 &  0.877$\pm$0.066 \\
    & $\mathcal{N}(0;1)$ & 0.1 & 10\% &  0.046$\pm$0.038 &  0.092$\pm$0.051 &  0.137$\pm$0.059 &  0.179$\pm$0.067 &  0.221$\pm$0.072 &  0.262$\pm$0.076 &  0.299$\pm$0.078 &  0.337$\pm$0.082 &  0.373$\pm$0.085 &  0.405$\pm$0.088 \\
    &                    & 0.3 & 10\% &  0.137$\pm$0.057 &  0.261$\pm$0.073 &  0.369$\pm$0.083 &  0.457$\pm$0.089 &  0.526$\pm$0.094 &  0.575$\pm$0.095 &  0.608$\pm$0.094 &  0.631$\pm$0.094 &  0.647$\pm$0.094 &  0.658$\pm$0.094 \\
100 & 0 & 0.1 & 10\% &  0.059$\pm$0.031 &  0.136$\pm$0.046 &  0.192$\pm$0.053 &  0.253$\pm$0.069 &  0.309$\pm$0.088 &  0.360$\pm$0.101 &  0.422$\pm$0.110 &  0.481$\pm$0.122 &  0.534$\pm$0.138 &  0.583$\pm$0.147 \\
    &                    & 0.3 & 10\% &  0.188$\pm$0.052 &  0.355$\pm$0.098 &  0.525$\pm$0.136 &  0.668$\pm$0.132 &  0.780$\pm$0.107 &  0.856$\pm$0.082 &  0.904$\pm$0.066 &  0.936$\pm$0.053 &  0.957$\pm$0.041 &  0.971$\pm$0.031 \\
    & $\mathcal{N}(0;0.08)$ & 0.1 & 10\% &  0.069$\pm$0.032 &  0.132$\pm$0.040 &  0.194$\pm$0.055 &  0.254$\pm$0.063 &  0.311$\pm$0.074 &  0.370$\pm$0.083 &  0.427$\pm$0.090 &  0.481$\pm$0.098 &  0.531$\pm$0.106 &  0.578$\pm$0.112 \\
    &                    & 0.3 & 10\% &  0.194$\pm$0.055 &  0.372$\pm$0.082 &  0.528$\pm$0.102 &  0.655$\pm$0.106 &  0.751$\pm$0.089 &  0.814$\pm$0.072 &  0.853$\pm$0.061 &  0.880$\pm$0.054 &  0.898$\pm$0.049 &  0.912$\pm$0.045 \\
    & $\mathcal{N}(0;1)$ & 0.1 & 10\% &  0.051$\pm$0.027 &  0.103$\pm$0.034 &  0.151$\pm$0.038 &  0.196$\pm$0.041 &  0.237$\pm$0.047 &  0.276$\pm$0.053 &  0.314$\pm$0.057 &  0.351$\pm$0.063 &  0.386$\pm$0.065 &  0.420$\pm$0.066 \\
    &                    & 0.3 & 10\% &  0.150$\pm$0.034 &  0.273$\pm$0.052 &  0.381$\pm$0.063 &  0.473$\pm$0.068 &  0.542$\pm$0.071 &  0.591$\pm$0.071 &  0.623$\pm$0.071 &  0.643$\pm$0.071 &  0.657$\pm$0.071 &  0.668$\pm$0.071 \\
\bottomrule
\end{tabular}

}
\caption{Results about the normalized $\Delta MoV$ of algorithm \emph{SPpagerank1.0\_pos} under the \emph{limited-knowledge} hypothesis for 20, 50, and 100 voters, 5 candidates, for single-peaked and nearly-single-peaked electorates. Each cell shows the average performance and the standard deviation. The target candidate of the simulated elections is the right-most one on the political spectrum.}
\label{tab:results_wattsstrogatz_LK_STUDYING_STEP_8_noise}
\end{table}
By analyzing the variances of the performances, we can see that standard deviations decrease when the number of voters increases, especially in noisy environments. This means that the variability of the solution tends to decrease compared to the size of the electorate.

\begin{figure}[ht]
% \vspace{-0.5cm}
    \centering
    \includegraphics[width=0.75\textwidth]{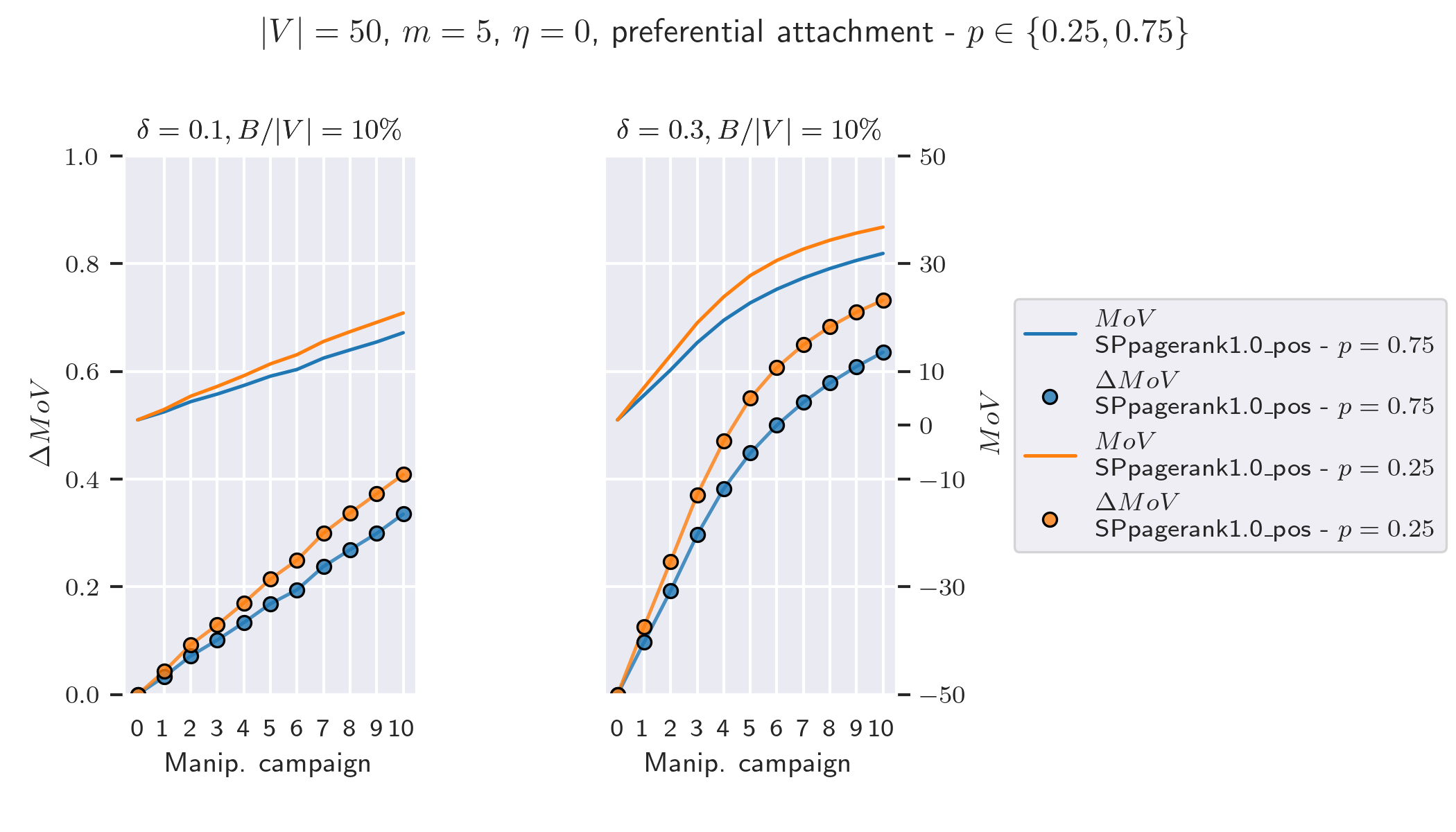}
    \caption{Performances of \emph{SPpagerank1.0\_pos} for $\NOISE=0$ and $p\in\{0.25,0.75\}$.}
    %The plots show both the normalized $\Delta MoV$ and the $MoV$.}
\label{fig:results_wattsstrogatz_LK_STUDYING_STEP_10}
% \vspace{-0.6cm}
\end{figure}
Next we evaluate whether results showed above are robust against different graphs. We first present the experiments that were performed on preferential attachment graphs.
% Exploiting the results of the previous sections, the fast heuristic considered for these tests is \emph{SPpagerank1.0\_pos}. However, since in preferential attachment graphs popularity and the number of followers are equivalent to the out-degree of a node in a preferential graph, we also considered the simple out-degree heuristic.
Tests were performed with $\NOISE=0$, $|V| \in \{20,50,100\}$, $\delta \in \{0.1, 0.3\}$.
% The probabilities required to build the graphs were $p=0.25$ (hence, a node links preferentially with probability 0.75) and $p=0.75$ (hence, a node links preferentially with probability 0.25).
The target candidate is the right-most one. Since results for different sizes of the electorates were almost identical,
% for brevity,
only the ones for $|V|=50$ are displayed.
Figure \ref{fig:results_wattsstrogatz_LK_STUDYING_STEP_10} displays the normalized $\Delta MoV$ and the $MoV$.
% The specific graph model (i.e., which probability $p$) is indicated in the legend next to the algorithm.
% There are two noteworthy results. First,
Observe that
the manipulator benefits from the rich-get-richer phenomenon.
% (details on
% % This makes sense: if the phenomenon is accentuated, then there are some people getting a lot of followers, and the manipulator simply has to find these nodes.
% % Second, algorithm \emph{SPpagerank1.0\_pos} is not that powerful when tested on this network model, being overwhelmed by the simple (and definitely faster) out-degree heuristic.
% % T
% the performances of our heuristics on these networks are given in Appendix~\ref{apx:preferential}).

Finally we show how our heuristics performs on the real Facebook network.
The results of the experiment are shown in Figure \ref{fig:facebook_mov}. Plots only show the margin of victory; $\Delta MoV$ can be plotted by simply shifting the curve up, such that the value before the first manipulation campaign is 0.
\begin{figure}
% \vspace{-0.3cm}
    \centering
    \includegraphics[width=0.7\textwidth]{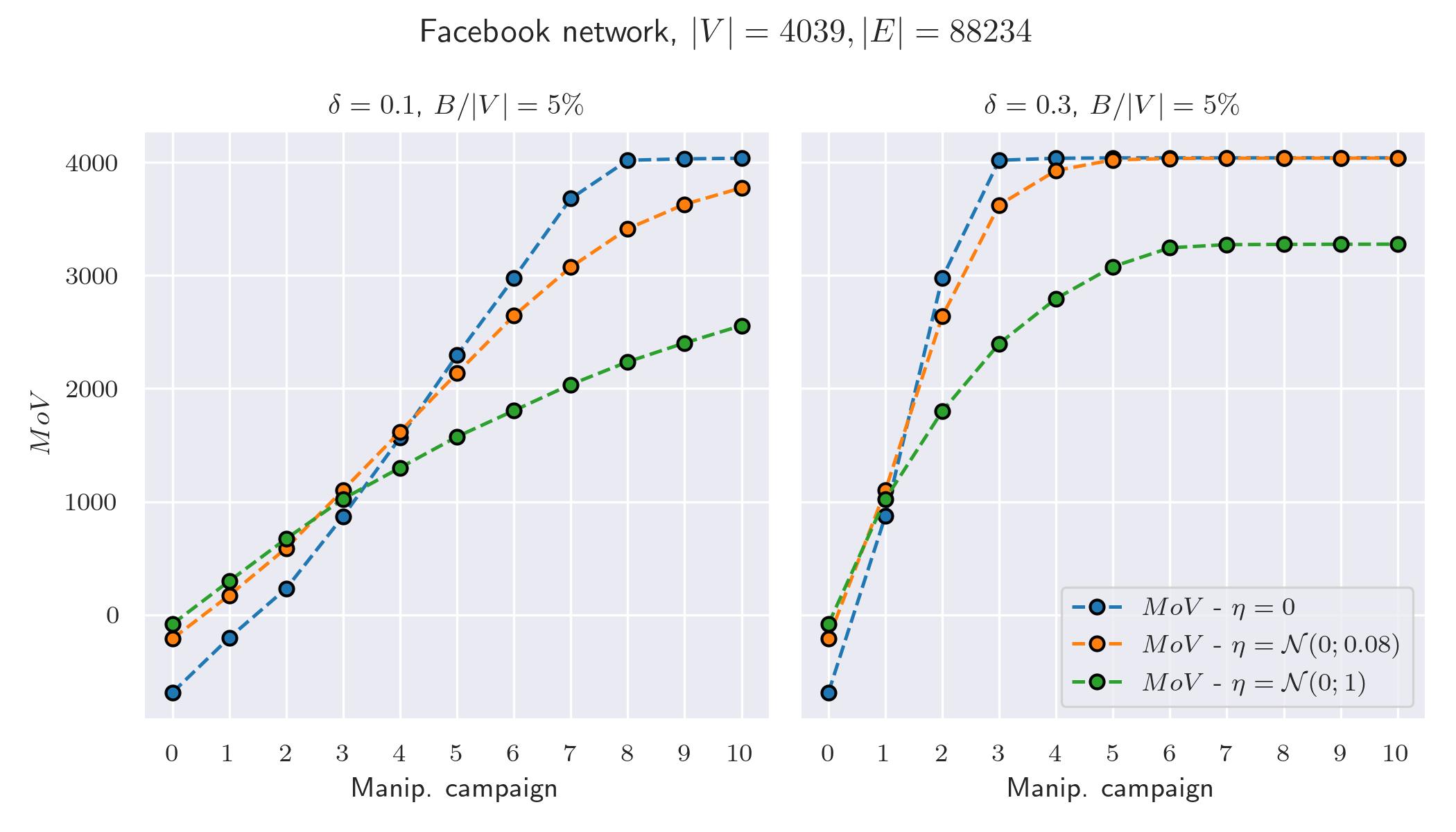}
    \caption{Average margin of victory for the test on the Facebook network.}
%     The algorithm is \emph{SPpagerank1.0\_pos}.}
\label{fig:facebook_mov}
% \vspace{-0.6cm}
\end{figure}
Note that candidates are in $\{-1, -0.5, 0, 0.5, 1\}$, and voters are initially placed such that $c_2$ (at position $0$) loses the election; in fact, his margin of victory is negative. Nevertheless, in single-peaked electorates, the algorithm only needs two campaigns (when $\delta=0.1$) or one campaign (when $\delta=0.3$) to make $c_2$ win the election.
% (the candidate wins the election when its margin of victory is positive).
Moreover, when voters are easily manipulable, the algorithm reaches unanimity in a few campaigns. Even in nearly-single-peaked electorates, the manipulator can make $c_2$ win, although performances are worse.

Visualizing communities as in Figure \ref{fig:facebook_louvain_communities} allows us to better understand how the algorithm works by analyzing the influencers it chooses for the manipulation.
\begin{figure}[H]
    \centering
    \includegraphics[width=0.8\textwidth]{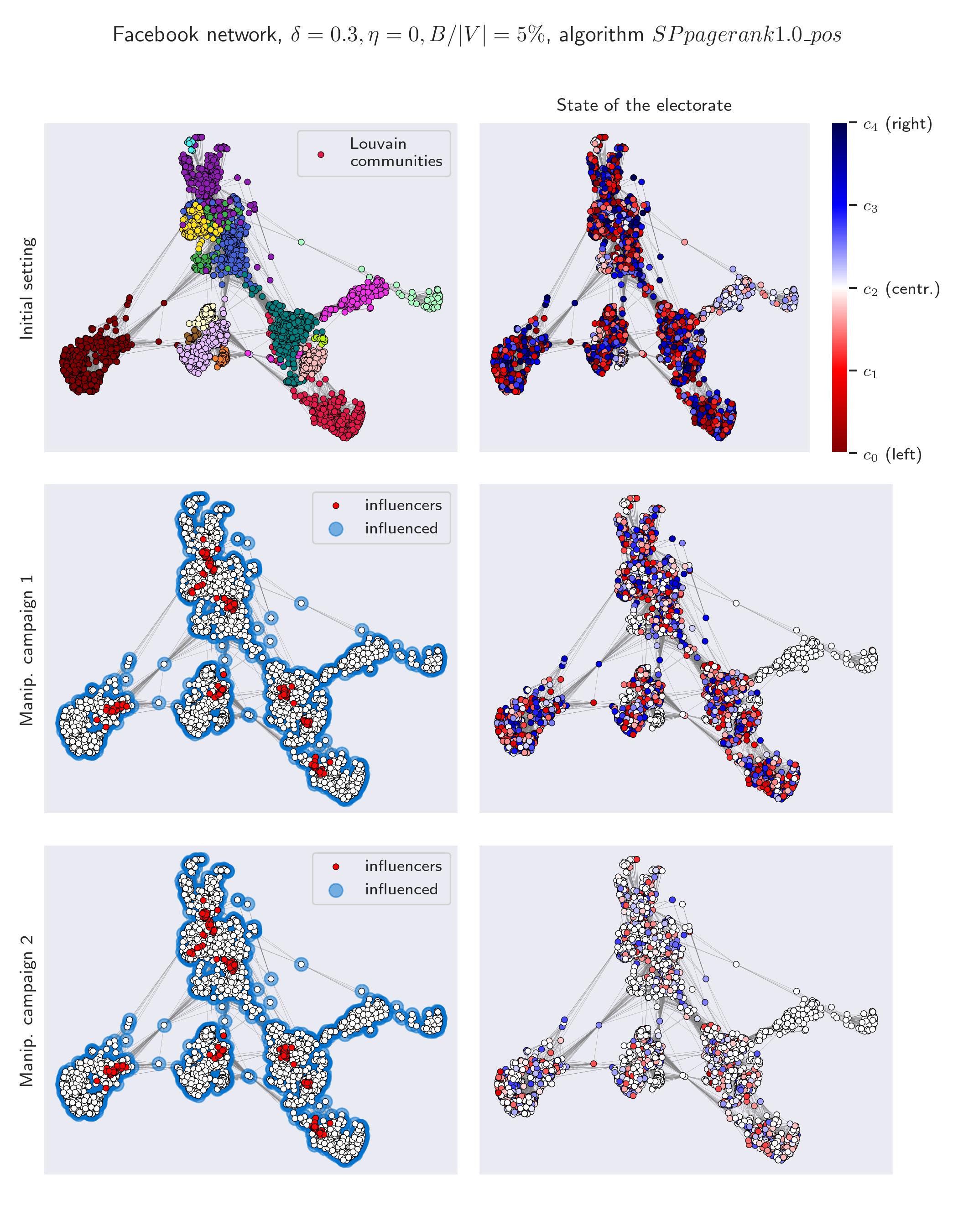}
    \caption{Facebook network. Evolution of the electorate after the first two manipulation campaigns.}
\label{fig:facebook_graph_manip}
\end{figure}
Remember that nodes in the same community are tightly connected, while voters in different communities have lower chances of sharing information with each other. Hence, the algorithm should spend the budget to choose nodes in different communities to reach as many nodes as possible. However, since this is an election manipulation problem and not an influence maximization one, the algorithm searches for focused influence, discarding nodes that already vote for the target $c_2$ as they do not impact the margin of victory. All these considerations are actually internalized by algorithm \emph{SPpagerank1.0\_pos}, as revealed in Figure \ref{fig:facebook_graph_manip}.
% spiegare la figure
The plots show the evolution of the electorate and the influencers chosen by the algorithm under the following conditions: $\NOISE=0$, $\delta=0.3$, $B/|V|=5\%$. Data concerns a single execution, as it would be in a real scenario: no averages were performed. The first row of the figure shows the Louvain communities as a reference and the initial state of the electorate (voters are coloured according to their political position on the spectrum; see the legend in the plot). It is clear that only a small subset of voters supports $c_2$, namely the communities on the right and some communities in the middle of the plot that are not clearly visible. The second row shows the influencers chosen by the algorithm and the state of the electorate after the first manipulation campaign. Interestingly, the influencers (represented as red dots) are distributed all over the communities to reach as many nodes as possible (represented as haloed, blue shadows). However, there is no influencer in the right part of the plot since such communities already supported $c_2$: choosing influencers in these communities would be a pointless waste of budget. The state of the electorate is represented by lighter colours which means that many voters changed their minds to support $c_2$. The same considerations apply to the second manipulation campaign (third row of the plot), after which much more nodes vote for the target candidate.
% , as expected
% (%
% % These results prove that the algorithm is very effective at manipulating the election, even when the target candidate is the least voted one.
% a more detailed analysis of the behavior of our heuristics on this network
% is given in Appendix~\ref{apx:facebook_test}).

We finally evaluate the timing performances (numerical results are showed in Table~\ref{tab:results_wattsstrogatz_LK_STUDYING_STEP_9_times}).
\begin{table}[H]
\centering
\resizebox{\textwidth}{!}{
% this comes from pandas
\begin{tabular}{llllllllllll}
\toprule
            & {} &     Round 1 &     Round 2 &     Round 3 &     Round 4 &     Round 5 &     Round 6 &     Round 7 &     Round 8 &     Round 9 &   Round 10 \\
Alg & {} &             &             &             &             &             &             &             &             &             &            \\
\midrule
$PR1^+$ & {} &   0.00242 s &   0.00235 s &   0.00228 s &   0.00221 s &   0.00214 s &   0.00207 s &   0.00202 s &   0.00197 s &   0.00194 s &  0.00192 s \\
            & {} &  $\pm$8.52e-05 &  $\pm$8.35e-05 &  $\pm$8.26e-05 &  $\pm$8.15e-05 &  $\pm$7.81e-05 &  $\pm$7.78e-05 &  $\pm$7.81e-05 &  $\pm$7.28e-05 &  $\pm$6.69e-05 &    $\pm$6e-05 \\ \hline
$APX$ & {} &      7.51 s &      4.09 s &       4.6 s &      3.65 s &      2.79 s &      1.81 s &      1.26 s &      0.89 s &     0.759 s &    0.693 s \\
            & {} &      $\pm$5.45 &      $\pm$2.37 &      $\pm$3.02 &       $\pm$2.5 &      $\pm$1.51 &     $\pm$0.886 &     $\pm$0.732 &     $\pm$0.672 &     $\pm$0.694 &    $\pm$0.715 \\
\bottomrule
\end{tabular}
}
\caption{Simulation times of the fast heuristics and the approximation algorithms for 20 voters, $\NOISE=0$, 5 candidates. Each cell shows the average execution time and the standard deviation in seconds. $PR1^+$ and $APX$ respectively stand for \emph{SPpagerank1.0\_pos} and \emph{Approximation algorithm}.}
\label{tab:results_wattsstrogatz_LK_STUDYING_STEP_9_times}
\end{table}
% % -EXP6: timing performances (4.16)
% \paragraph{Timing Performances.}
% This section reports the timing tests; they are useful to understand if an evil manipulator is really capable of running the algorithms on real instances of the manipulation problem.
%
% Results in Table \ref{tab:results_wattsstrogatz_LK_STUDYING_STEP_9_times} finally explain why the number of tests of the approximation algorithm was so dramatically reduced.
% % (the related experiments are the same as the ones in Figure \ref{fig:results_wattsstrogatz_LK_STUDYING_STEP_9}).
%
Execution times show that the approximation algorithm is thousands of times slower than the fast heuristic.
% Of course, these outcomes are not totally comparable since \emph{SPpagerank1.0\_pos} is based on optimized code to compute the PageRank while the approximation algorithm is implemented in pure Python. However, these results justify the fact that only a few tests were performed for the approximation algorithm. Moreover, now it is clear that the manipulator must be equipped with very powerful machines to run this algorithm that requires seconds to manipulate a simple 20-voter electorate. Certainly, one could reduce the precision to compute the average number of influenced nodes, but this would also mean a low-quality solution.
% It all depends on the computational capabilities of the manipulator and, in case of limited resources, one could use the proposed fast heuristic \emph{SPpagerank1.0\_pos}.
% As a final note, the higher the round number, the lower the execution times. This is due to the fact that the electorate could be totally convinced to vote for $c^*$ at a certain campaign $j$. In that case, the simulations assumed that the execution times of campaigns $i>j$ is 0 since there is no one left to manipulate; this reduces the average execution times.

% fare vedere lo stress test con e senza rumore. con rumore faccio meno tests perché l'elettorato arriva più tardi alla totalità a votare c^*, e questo conferma che il rumore fa andare peggio la manipolazione
Finally, we tested the scalability of \emph{SPpagerank1.0\_pos} on networks up to 20000 nodes as described above
% , the next experiment repeatedly runs the algorithm for three hours on random electorates of increasing sizes, recording the total number of completed simulations. The graphs are random Watts-Strogatz models built to keep the average degree of the nodes constant. The target candidate is the right-most on the political spectrum. The budget is 10\% of the size of the electorate; $\delta=0.3$.
(the numerical results of this experiment are shown in Table \ref{tab:stress_test}).
\begin{table}[H]
\centering
\begin{tabular}{l|rrrrrrr}
\toprule
{$\NOISE$} &  $|V|=200$ &  $500$ &  $1000$ &  $2000$ &  $5000$ &  $10000$ &  $20000$ \\
\midrule
0                   &      14743 &       7558 &        3842 &        1613 &         210 &           59 &           15 \\
$\mathcal{N}(0;0.08)$ &      13204 &       6848 &        3495 &        1490 &         193 &           56 &           14 \\
$\mathcal{N}(0;1.0)$  &      12693 &       6699 &        3443 &        1473 &         190 &           55 &           15 \\
\bottomrule
\end{tabular}
\caption{Number of simulations of algorithm \emph{SPpagerank1.0\_pos} running on electorates of increasing sizes for three hours.}
\label{tab:stress_test}
\end{table}
Interestingly, on a common PC, the algorithm can be executed 15 times in three hours on graphs of 20000 nodes.
% Moreover,
Since simulations require additional code to prepare the electoral setting and include 10 manipulation campaigns, the number of tests runnable in three hours is even higher. This means that a manipulator would not face any problem executing the proposed algorithm on large graphs.
% representing real societies. In fact, he is supposed to be equipped with powerful machines able to do much more than a simple computer. Finally, we can see that the stronger the noise, the lower the number of simulations. In fact, in perfectly single-peaked electorates, the manipulator convinces voters more easily, being able to reach unanimity and stop his activity by exploiting less than 10 manipulation rounds.
% A more thorough analysis about the bottleneck of the running time of our heuristics can be found in Appendix~\ref{apx:time_analysis}.

The plot in Figure \ref{fig:results_subtimes} analyzes the same results of Table \ref{tab:stress_test} from a different perspective: it shows the average values of the total execution times, the execution times to compute the scores $z(v)$ of the nodes $v$, and the execution times to compute the weighted PageRank based on such scores. The plot considers only the first manipulation campaign since related execution times are surely not altered, unless unanimity in favour of the target candidate is reached (and thus the simulation is stopped).
\begin{figure}[H]
    \centering
    \includegraphics[width=1.0\textwidth]{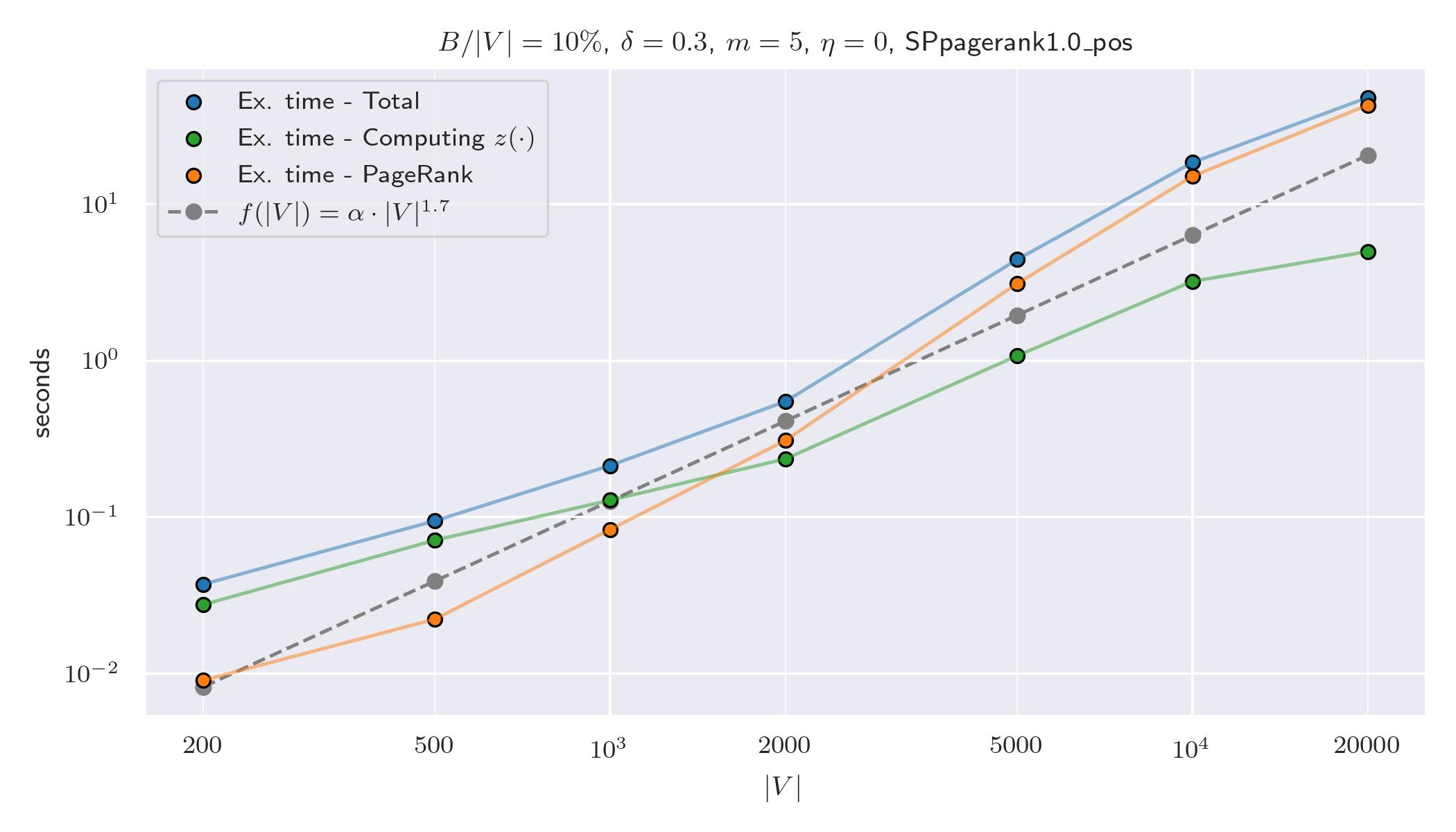}
    \caption{Execution times of the algorithms \emph{SPpagerank1.0\_pos} in perfectly-single-peaked scenarios. The plot shows the average execution times to run the algorithm, compute the weights $z(\cdot)$, and compute PageRank values. The dashed line confirms that execution times grow according to a power law.}
\label{fig:results_subtimes}
\end{figure}
We can see that the cost of the algorithm is basically the one of PageRank, as the cost to compute $z(\cdot)$ is negligible. We know that search engines use it constantly on graphs of millions of nodes. Therefore, the algorithm is surely runnable on real problem instances involving millions of voters.
% log log plot per vedere la curva
The plot intentionally uses logarithmic scales for both the x-axis and the y-axis. In fact, it is clear that the execution times of the PageRank algorithm approximately follow a straight line proving that execution times grow according to a power law.
% As explained in section \ref{sec:preferential_graph}, this suggests that the curve is a power law: a good candidate is $f(|V|)=\alpha \cdot |V|^{1.7}$, where $\alpha$ is a constant that depends on the speed of the machine and is not that important (in this case, $\alpha=10^{-6}$). Considering that (1) we are dealing with an optimized implementation, (2) the simple PageRank algorithm provided in the previous chapter runs in $O(k(|E|+|V|))$ time and (3) $|E|=O(|V|^2)$ for any graph, these results perfectly make sense.

\section{Conclusion}
In this work we considered the problem of election manipulation through social influence when agents have single-peaked or nearly single-peaked preferences.
For this purpose, we first propose a new manipulation model that intrinsically generates single-peaked preferences of the voters.
We provided an algorithm with constant approximation guarantees whenever there is no agent that is more advantaged than the target candidate by a campaign in favour of the latter. We also provided an heuristics that has been proved to perform very well in simulations and to be computable very efficiently.
These results highlight the huge risk of election manipulation in the single-peaked setting.

It would be desirable to further extend and deepen our analysis. Moreover, it would be interesting to design efficient and effective counter-measures against manipulation. Our analysis, by highlighting those aspects that simplify or complicate the manipulation, may be an useful starting point in this direction.
%
% In conclusion, by analyzing the results of the experiments, this work proposes countermeasures against manipulation. In particular, to limit the power of an evil party, the main authority ruling the election should limit the budget on the allowed propaganda of the candidates. Moreover, voters' education and reluctance to believe the news they hear about the parties are relevant factors, too. Finally, results prove that manipulation is easier when detailed information about the electorate is available to the evil party. Therefore, people should not publicly announce their political ideas: any information in this sense could help the manipulator reach his goal. However, this is a subtle aspect: nowadays, any activity (such as shares, likes, or participation in political blogs) on an online social network can be tracked by an evil party and exploited to build detailed political profiles of the voters.
\pagebreak

%\ack We would like to thank the referees for their comments, which helped improve this paper considerably

\bibliographystyle{splncs04}
\bibliography{carbone}

\end{document}